\documentclass[review,12pt]{elsarticle}

\usepackage{amsmath,amssymb}

\usepackage{subcaption}

\graphicspath{{img/}}

\newcommand{\pre}[1]{{^\bullet{#1}}}
\newcommand{\post}[1]{{#1}^\bullet}
\newcommand{\neighb}[1]{{^\bullet{#1}^\bullet}}
\newcommand{\inp}[1]{{^\bigcirc{#1}}}
\newcommand{\outp}[1]{#1^\bigcirc}
\newcommand{\reach}[1]{\lbrack#1\rangle}
\newcommand{\om}[1]{\varphi(#1)}
\newcommand{\omr}[1]{\varphi^{-1}(#1)}
\newcommand{\abs}[1]{\lvert #1 \rvert}

\newcommand{\nat}{\mathbb{N}}
\newcommand{\chan}{\mathfrak{C}}

\newcommand{\unf}[1]{\mathcal{U}(#1)}
\newcommand{\dom}[1]{\text{dom}(#1)}
\newcommand{\rng}[1]{\text{rng}(#1)}

\newtheorem{definition}{Definition}
\newtheorem{proposition}{Proposition}
\newtheorem{theorem}{Theorem}
\newtheorem{lemma}{Lemma}
\newtheorem{corollary}{Corollary}
\newproof{proof}{Proof}

\begin{document}

\begin{frontmatter}



\title{Soundness-Preserving Composition of Synchronously and Asynchronously Interacting Workflow Net Components}

\author[label1]{Luca Bernardinello}
\author[label2]{Irina Lomazova}
\author[label1,label2]{Roman Nesterov}
\author[label2]{Lucia Pomello}

\affiliation[label1]{organization={University of Milano-Bicocca},
            addressline={Viale Sarca 336 - Edificio U14}, 
            city={Milan},
            postcode={20126}, 
            country={Italy}}
        
\affiliation[label2]{organization={HSE University},
     	addressline={11 Pokrovskiy Boulevard}, 
     	city={Moscow},
     	postcode={101000}, 
     	country={Russia}}

\begin{abstract}
In this paper, we propose a compositional approach to constructing correct formal models of information systems from correct models of interacting components.
Component behavior is represented using workflow nets --- a class of Petri nets.
Interactions among components are encoded in an additional interface net. 
The proposed approach is used to model and compose synchronously and asynchronously interacting workflow nets.
Using Petri net morphisms and their properties, we prove that the composition of interacting workflow nets preserves the correctness of components and of an interface.

\end{abstract}



\end{frontmatter}


\section{Introduction}
\label{}

Formal models are essential for the specification and analysis of a distributed information system behavior. 
The precise semantics of such models helps to prove various important properties, which concern reliability and smooth operation of information systems. 
\emph{Petri nets} \cite{Reisig13} are widely recognized as one of the most convenient formalisms for modeling and analyzing the behavior of complex distributed systems. 

Petri net composition has been extensively studied in the literature. 
Researchers considered various aspects, including \emph{architectural} concepts of compositional modeling and \emph{semantical} issues relating to compositional analysis of Petri net behavior. 
The ubiquity of service-oriented and multi-agent architectures of information systems retains the relevance of further research on these aspects of Petri net composition.

Wolfgang Reisig, in his recent works \cite{Reisig2018,ReiComp20}, defined a general setting for compositional modeling of service-oriented information systems. 
He addressed \emph{architectural} problems behind the composition of Petri net components in the context of algebraic properties.

The compositional analysis of \emph{semantical} aspects has been performed for different classes of Petri nets and behavioral properties. 
Among works in this direction, we note the one \cite{Wolf09} by Christian Stahl and Karsten Wolf, who studied the compositional analysis of deadlock-freeness in open Petri nets.

Boundedness and deadlock-freeness are two components of \emph{soundness} --- the crucial correctness property of \emph{workflow} (WF) nets \cite{Aalstwf02}. 
WF-nets form a class of Petri nets used to model the \emph{control-flow} of processes in information systems. 
The method proposed by C.\,Stahl and K.\,Wolf does not consider the compositional proof of the absence of livelocks --- the third component of soundness.

Our study is focused on semantical aspects of compositional WF-net modeling, namely, the construction of sound WF-nets from sound models of interacting components.
The following example illustrates that unregulated interactions of sound WF-nets can easily violate soundness. Figure \ref{in_ex} shows two WF-net components $N_1$ and $N_2$. 
They both are sound with respect to their initial and final states, where $s_1$ ($s_2$) is the initial state in  $N_1$ ($N_2$), and $f_1$ ($f_2$) is the final state in $N_1$ ($N_2$).

\begin{figure}[h]
	\centering
	\includegraphics[height=4.5cm]{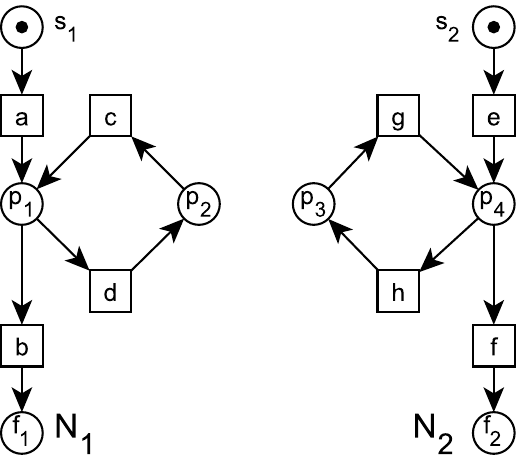}
	\caption{Two sound WF-net components}\label{in_ex}
\end{figure}

Let us suppose that $N_1$ and $N_2$ interact \emph{synchronously}, i.e., execute a simultaneous activity. 
Let transitions $c$ and $g$ correspond to this activity. 
Thus, we need to fuse transitions $c$ and $g$, preserving the original arcs. 
Figure \ref{in_ex21} shows the result, where the merged transition is denoted by $(c, g)$. 
This synchronization produces a deadlock since a marking with tokens in $p_2$ and $p_3$ is not always reachable.

Next, let us suppose that $N_1$ and $N_2$ interact \emph{asynchronously}, i.e., exchange messages through channels. 
Let $N_1$ send messages via transition $d$, and let $N_2$ receive messages via transition $h$. 
Then we add a place $m$ to model a channel between transitions $d$ and $h$. 
We need to connect these transitions with the added place according to sending and receiving operations. 
Figure \ref{in_ex22} shows the insertion of a place between transitions $d$ and $h$. 
This asynchronous interaction leads to potential overflow in the added place $m$. 
Therefore, this composition of $N_1$ and $N_2$ is unbounded.

\begin{figure}[h]
	\centering
	\begin{subfigure}[b]{0.5\textwidth}
		\centering
		\includegraphics[height=4.5cm]{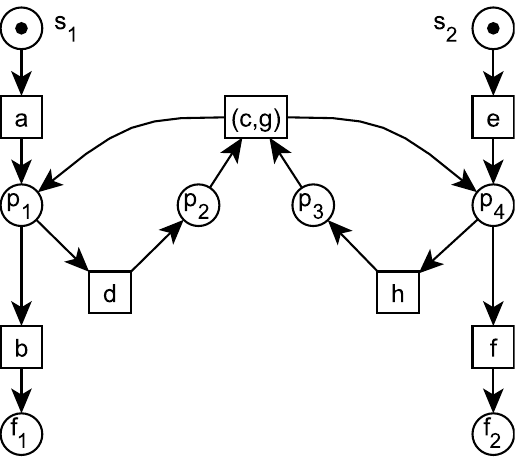}
		\caption{synchronous\label{in_ex21}}
	\end{subfigure}%
	\begin{subfigure}[b]{0.5\textwidth}
		\centering
		\includegraphics[height=4.5cm]{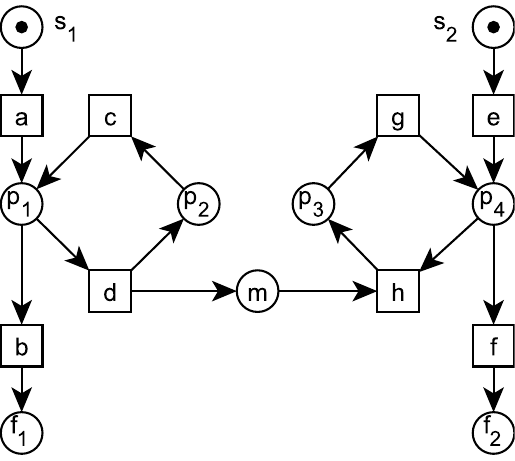}
		\caption{asynchronous\label{in_ex22}}
	\end{subfigure}
	\caption{Interactions between the WF-net components from Fig. \ref{in_ex}}
\end{figure}

Our study discusses theoretical backgrounds to justify a correct composition of interacting workflow nets.
We define an operation of composing synchronously and asynchronously interacting WF-nets. 
This composition is defined on a syntactical level and does not guarantee to preserve soundness.
We use two types of components in compositional modeling: agent models and an interface that describes how agents interact. 
Both agent and interface models are sound WF-nets.

An interface plays a significant role in formulating conditions for a semantically correct composition of sound WF-nets. 
An interface represents the \emph{abstract} view of a complete system, where the behavior of an agent corresponds to a subnet describing requirements imposed on the structure of an agent.
The correspondence between an agent behavior and an interface subnet is defined through an abstraction/refinement relation based on $\alpha$-morphisms \cite{Bernardinello2013}. 
We prove that replacing an interface subnet with the corresponding agent behavior preserves interface soundness.

Thus, the main contributions of our paper are:
\begin{enumerate}
	\item Formal definition and semantical properties of asynchronous-synchronous composition of interacting WF-nets.
	\item Structural and behavioral properties of place refinement and subnet abstraction in WF-nets based on $\alpha$-morphisms.
	\item Proof of the theorem that component refinement in asynchronous-synchronous composition of WF-nets preserves soundness.
\end{enumerate}

The remainder of the paper proceeds as follows.
The next section provides basic definitions on Petri nets, workflow nets, and their behavior.
In Section~\ref{sec:comp}, we define an asynchronous-synchronous workflow net composition and study its properties.
In Section~\ref{sec:alpha}, using $\alpha$-morphisms, we define an abstraction/refinement relation on workflow nets and study the relevant properties of this relation.
Section~\ref{sec:main} describes how to preserve the properties of workflow net components in their asynchronous-synchronous composition through the use of the abstraction/refinement relation.
Section~\ref{sec:relw} gives a review of related works, and Section~\ref{sec:concl} concludes the paper. 


\section{Preliminaries}\label{sec:prel}

This section provides the basic definitions on Petri nets used in the paper.

Let $A, B$ be two sets.
A function $f$ from $A$ to $B$ is denoted by $f \colon A \to B$, where $A$ is the \emph{domain} of $f$ (denoted by $\dom{f}$) and $B$ is the \emph{range} of $f$ (denoted by $\rng{f}$).
A \emph{restriction} of a function $f$ to a subset $A' \subseteq A$ is denoted by $f \vert_{A'} \colon A' \to B$.
A \emph{partial} function $g$ from $A$ to $B$ is a function from $A'$ to $B$, where $A' \subseteq A$.
A partial function is denoted by $g \colon A \nrightarrow B$.
When $g$ is not defined for $a \in A$, we write $g(a) = \perp$.

Let $\nat$ denote the set of non-negative integers. 
A \textit{multiset} $m$ over a set $S$ is a function $m \colon S \rightarrow \nat$.
If $m(s)\geq 1$, we write $s \in m$.
If $m(s) \leq 1$ for all $s \in S$, then $m$ corresponds a set $S' \subseteq S$.

Let $m_1, m_2$ be two multisets over the same set $S$. 
Then $m_1 \subseteq m_2 \Leftrightarrow m_1(s)\leq m_2(s)$ for all $s \in S$. 
Also, $m'=m_1+m_2 \Leftrightarrow m'(s)=m_1(s)+m_2(s)$, $m''=m_1-m_2 \Leftrightarrow m''(s)=\max(m_1(s)-m_2(s), 0)$ for all $s\in S$.

Let $A^+$ denote the set of all finite non-empty \emph{sequences} over $A$, and $A^* = A^+\cup \{\varepsilon\}$, where $\varepsilon$ is the empty sequence.
Then for $w \in A^*$ and $B \subseteq A$, $w \vert_{B}$ denotes the \emph{projection} of $w$ on $B$, i.e., $w\vert_{B}$ is the sub-sequence of $w$ built from elements in $B$.
For example, if $A = \{a, b, c\}$, $B=\{b\}$, and $w = aabbbcc \in A^*$, then $w\vert_{B} = bbb$.
\medskip

A \emph{Petri net} is a triple $N=(P, T, F)$, where $P$ and $T$ are two disjoint sets of places and transitions, i.e., 
$P \cap T = \varnothing$, and $F \subseteq (P \times T) \cup (T \times P)$ is a flow relation.
Pictorially, places are shown by circles, transitions are shown by boxes, and $F$ is shown by arcs.
For $N_1$, shown in Fig. \ref{in_ex}, $P = \{s_1, p_1, p_2, f_1\}$, $T = \{a, b, c, d\}$ and $F=\{(s_1, a), (a, p_1), (p_1, b), (p_1, d), (d, p_2), (p_2, c), (c, p_1),$ $(b, f_1)\}$.

Let $N=(P, T, F)$ be a Petri net, and $X = P \cup T$. 
The set $\pre{x} = \{y \in X \,\vert\, (y, x) \in F\}$ is called the \textit{preset} of $x \in X$. 
The set $\post{x} = \{y \in X \,\vert\, (x, y) \in F\}$ is called the \textit{postset} of $x \in X$.
The set $\neighb{x} = \pre{x} \cup \post{x}$ is called the \textit{neighborhood} of $x \in X$.
$N$ is \emph{P-simple} iff $\forall p_1, p_2 \in P \colon \pre{p_1} = \pre{p_2}$ and $\post{p_1} = \post{p_2}$ implies $p_1 = p_2$.
In our study, we consider Petri nets, such that $\nexists x \in X \colon \pre{x} = \varnothing = \post{x}$ and $\forall t \in T \colon \abs{\pre{t}} \geq 1$ and $\abs{\post{t}} \geq 1$.
Self-loops are forbidden, i.e., $\forall x \in X \colon \pre{x} \cap \post{x} = \varnothing$.

Let $N=(P, T, F)$ be a Petri net, and $A \subseteq X$. 
Then $\pre{A} = \bigcup_{x \in A}\pre{x}$, $\post{A} = \bigcup_{x \in A}\post{x}$, $\neighb{A} = \pre{A} \cup \post{A}$.
Let $N(A)$ denote a \textit{subnet} of $N$ \textit{generated by} $A$, i.e., $N(A) = (P\cap A, T \cap A, F \cap (A \times A))$. 
The set $\inp{N(A)} = \{y \in A \,\vert\, (\exists z \in X \setminus A \colon (z, y) \in F) \text{ or } (\pre{y} = \varnothing)\}$ contains the \textit{input} elements, 
and the set $\outp{N(A)} = \{y \in A \,\vert \,(\exists z \in X \setminus A \colon (y, z) \in F) \text{ or } (\post{y} = \varnothing)\}$ contains the \emph{output} elements of the subnet $N(A)$.
\medskip

A \emph{marking} (state) in a Petri  net $N=(P, T, F)$ is a multiset over $P$. 
A marking $m$ is designated by putting $m(p)$ black tokens inside a place $p \in P$.
Transition $t \in T$ has a \emph{contact} at a marking $m$ if $\pre{t} \subseteq m$ and $\post{t} \cap m \neq \varnothing$. 
A \emph{marked} Petri net is a quadruple $N=(P, T, F, m_0)$, where $(P, T, F)$ is a Petri net and  $m_0$ is the \emph{initial} marking.
Further, the term ``marked'' can be omitted while referring to marked Petri nets.
\medskip

A \emph{state machine} is a connected Petri net $N =(P, T, F)$, where $\forall t \in T\colon \abs{\pre{t}}=\abs{\post{t}}=1$. 
A subnet of a Petri net $N = (P, T, F, m_0)$ identified by a subset of places $A \subseteq P$ and its neighborhood, i.e.,
$N(A \, \cup \, ({\neighb{A}}))$, is a \emph{sequential component} of $N$ iff it is a state machine and has a single token in the initial marking. 
$N$ is \emph{covered} by sequential components if every place in $N$ belongs to at least one sequential component of $N$. 
In this case, $N$ is said to be \emph{state machine decomposable} (SMD).

For instance, subnet $N(A \, \cup \, ({\neighb{A}}))$ identified by the set of places $A = \{s_1, p_1, p_2, f_1\}$ and by the corresponding set of transitions $\neighb{A} = \{a, b, c, d\}$ is a sequential component of the Petri net shown in Fig. \ref{in_ex22}.
However, this Petri net is not state machine decomposable, since there is no sequential component containing place $m$.

\medskip

The behavior of a Petri net is defined according to the \emph{firing rule}, which specifies when a transition may fire and how a state in a Petri net changes.

A marking $m$ in $N=(P, T, F, m_0)$ \emph{enables} a transition $t \in T$, denoted $m\reach{t}$, if $\pre{t} \subseteq m$.
When $t$ \emph{fires}, $N$ evolves to a new marking $m' = m - \pre{t} + \post{t}$.
We write $m \reach{t} m'$.
A sequence $w \in T^*$ is a \emph{firing sequence} of $N =(P, T, F, m_0)$ iff $w=t_1t_2\dots t_n$ and $m_0\reach{t_1} m_1 \reach{t_2}\dots m_{n-1}\reach{t_n} m_n$. Then we write $m_0\reach{w}m_n$. The set of all firing sequences of $N$ is denoted by $F\!S(N)$. 

A marking $m$ in $N =(P, T, F, m_0)$ is \emph{reachable} if $\exists w \in F\!S(N) \colon m_0\reach{w}m$. 
The set of all markings in $N$ reachable from $m$ is denoted by $\reach{m}$. 
$N$ is \emph{safe} iff $\forall p \in \!P$, $\forall m \in \reach{m_0} \colon m(p) \leq 1$. 
Thus, a reachable marking in a safe Petri net is a set of places.

Petri nets covered by sequential components are free of contacts \cite{Rozenberg96}.
That is why Petri nets covered by sequential components are safe.

\medskip

The \emph{concurrent} semantics of a Petri net is captured by its \emph{unfolding}.

Let $N = (P, T, F)$ be a Petri net, and $F^*$ be the reflexive transitive closure of F. Then $ \forall x, y \in P \cup T \colon$
$x$ and y are in \emph{causal} relation, denoted $x \leq y$, if $(x, y) \in F^*$; $x$ and $y$ are in \emph{conflict} relation, denoted $x \# y$, if $\exists t_x, t_y \in T$, such that  $t_x \neq t_y$, $\pre{t_x} \cap \pre{t_y} \neq \varnothing$, and $t_x \leq x$, $t_y \leq y$.
\begin{definition}
	A Petri net $O = (B, E, F)$ is an occurrence net iff:
	\begin{enumerate}
		\item $\forall b \in B \colon \abs{\pre{b}} \leq 1$.
		\item $F^*$ is a partial order.
		\item $\forall x \in B \cup E \colon \{y \in B \cup E \,\vert\, y \leq x\}$ is finite.
		\item $\forall x, y \in B \cup E \colon x \# y \Rightarrow x \neq y$.
	\end{enumerate}
\end{definition}

By definition, $O$ is acyclic. 
Let $M\!in(O)$ denote the set of minimal nodes of $O$ w.r.t. $F^*$, i.e., the elements with the empty preset. 
Since we consider nets having transitions with non-empty presets and postsets, $M\!in(O) \subseteq B$.

\begin{definition}
	Let $N=(P, T, F, m_0)$ be a safe Petri net, $O=(B, E, F)$ be an occurrence net, and $\pi:B\cup E\to P\cup T$ be a map. A couple $(O, \pi)$ is a branching process of $N$ iff:
	\begin{enumerate}
		\item $\pi(B) \subseteq P$ and $\pi(E) \subseteq T$.
		\item $\pi \vert_{M\!in(O)}$ is a bijection from $M\!in(O)$ to $m_0$.
		\item $\forall t \in T\colon \pi \vert_{\pre{t}}$ is a bijection between $\pre{t}$ and $\pre{\pi(t)}$, and similarly for $\post{t}$ and $\post{\pi(t)}$.
		\item $\forall t_1, t_2 \in T \colon$ if $\pre{t_1} = \pre{t_2}$ and $\pi(t_1) = \pi(t_2)$, then $t_1 = t_2$.
	\end{enumerate}
\end{definition}

The \emph{unfolding} of $N$, denoted $\mathcal{U}(N)$, is the maximal branching process of $N$, such that any other branching process of $N$ is isomorphic to a subnet of $\mathcal{U}(N)$, where the map $\pi$ is restricted to the elements of this subnet. 
The map associated with the unfolding is denoted $u$ and called \emph{folding}.
\medskip

\emph{Workflow nets} form a subclass of Petri nets used for modeling processes and services.
They have unique input and output places.
We define a \emph{generalized} workflow net with the initial state $m_0$ (exactly corresponding to its initial marking) and the final state $m_f$ below.

\begin{definition}\label{GWF}
	A generalized workflow (GWF) net $N = (P, T, F, m_0, m_f)$ is a Petri net $(P, T, F,$ $m_0)$ equipped with $m_f$, where:
	\begin{enumerate}
		\item $\forall p \in m_0 \colon \pre{p} = \varnothing$. \label{instwf}
		\item $m_f \subseteq P$ such that $m_f \neq \varnothing$ and $\forall p \in m_f \colon \post{p} = \varnothing$.\label{outstwf}
		\item $\forall x \in P \cup T \,\,\exists s \in m_0 \,\,\exists f \in m_f \colon (s, x) \in F^* \text{ and } (x, f) \in F^*$.\label{connwf}
	\end{enumerate}
\end{definition}

Since SMD Petri nets are safe, GWF-nets covered by sequential components are safe as well.
The important correctness property of GWF-nets is \emph{soundness} \cite{Aalst11} formally defined below.

\begin{definition}\label{sound}
	A GWF-net $N = (P, T, F, m_0, m_f)$ is sound iff:
	\begin{enumerate}
		\item $\forall m \in \reach{m_0} \colon m_f \in \reach{m}$.\label{sndprop}
		\item $\forall m \in \reach{m_0} \colon m_f \subseteq m \Rightarrow m=m_f$.\label{sndclean}
		\item $\forall t \in T \, \exists m \in \reach{m_0} \colon m\reach{t}$.\label{sndlive}
	\end{enumerate}
\end{definition}

Soundness is threefold. 
Firstly, the final state in a sound GWF-net is reachable from any reachable state (\emph{proper termination}). 
Secondly, the final state in sound GWF-net cannot be contained in any other reachable state (\emph{clean termination}). 
Finally, each transition in a sound GWF-net can fire.

\section{Asynchronous-Synchronous Composition of GFW-Nets}\label{sec:comp}

In this section, we develop an approach to modeling synchronous and asynchronous interactions among components in a system.
Components are represented using GWF-nets.
We introduce \emph{transition labels} and a corresponding \emph{AS-composition} which merges synchronous transitions and adds channels between asynchronously interacting transitions in component models.
Some basic properties of the AS-composition are also studied here.

\subsection{Labeled GWF-Nets}

We introduce two kinds of transition labels to model the \emph{asynchronous} and \emph{synchronous} interaction among system components.
Their behavior is modeled with the help of GWF-nets covered by sequential components.
\medskip

When components interact asynchronously, they exchange messages using \emph{channels}.
Correspondingly, components can \emph{send} (\emph{receive}) messages \emph{to} (\emph{from}) channels.
Let $\chan = \{c_1, c_2, \dots, c_k \}$ denote the set of all channels.
Channels are represented by places. 
The set $\Lambda$ of sending/receiving actions implemented over channels  is defined as $\Lambda = \{c!, c? \, \vert \, c \in \chan \}$, where ``$c!$'' indicates sending a message to a channel $c$, and ``$c?$'' indicates receiving a message from a channel $c$.
Thus, some transitions in a GWF-net are labeled by asynchronous actions from $\Lambda$.

For actions in $\Lambda$, we define a function $\mathbf{ch}\colon \Lambda \to \chan$, which maps a sending/receiving action to a corresponding channel, i.e. $\mathbf{ch}(c!) = \mathbf{ch}(c?) = c$. 
Given $\Lambda' \subseteq \Lambda$, $\mathbf{ch}(\Lambda') = \bigcup_{\lambda' \in \Lambda'} \mathbf{ch}(\lambda')$.

A GWF-net may also have transitions with \emph{complement} asynchronous labels (``$c!$'' is complement to ``$c?$'' and vice versa).
They are denoted using overline, i.e., $\overline{c!} = c?$ and $\overline{c?} = c!$.
Then we require that there exists a place labeled by ``$c$'' connecting all transitions labeled by ``$c!$'' to all transitions labeled by ``$c?$''.
Labeled places are necessary to establish the logical dependence between transitions with complement labels, i.e., receiving from a channel $c$ should be done after a message is sent to this channel.
However, other places can also connect transitions with complement labels.
\medskip

Synchronous interactions among components result in merging transitions representing simultaneous actions.
In our work, simultaneous actions are modeled using identical transition labels.
Let $S = \{ s_1, s_2, \dots, s_n \}$ denote the set of synchronous actions.
Similar to the asynchronous interaction, some transitions in a GWF-net are labeled by synchronous actions from $S$.
\medskip

We formalize these aspects of synchronous and asynchronous interactions among components in Definition \ref{lwf_def}, where a GWF-net is equipped with two transition labeling functions and a place labeling function.
Figure \ref{lgwf_ex} shows a labeled GWF-net, where labeled places are distinguished by the smaller size.
By convention, labels are put either inside or near nodes.

\begin{definition}\label{lwf_def}
	Let $\chan$ be a set of channels, and $\Lambda = \{c!, c? \,\vert\, c \in \chan\}$ be a set of sending/receiving actions over $\chan$.
	Let $S$ be a set of synchronous actions.
	A labeled GWF-net (LGWF-net) $N = (P, T, F, m_0, m_f, h, \ell, k)$ is a GWF-net $(P, T,$ $F, m_0, m_f)$ with two transition labeling functions $h$, $\ell$ and a place labeling function $k$, where:
	\begin{enumerate}
		\item $h \colon T \nrightarrow \Lambda \text{ is a partial function}$.
		\item $\ell \colon T \nrightarrow S$ is a partial function, $\dom{h} \cap \dom{\ell} = \varnothing$.
		\item $k \colon P \nrightarrow \chan$ is a partial injective function, such that: \label{chanlab}
		\begin{enumerate}
			\item $\forall t_1,t_2 \in T \colon$  if
			$h(t_1) = c!$ and $h(t_2)=c?$, then  \\ $\exists p\in P \colon k(p) = c$ and $(t_1, p), (p, t_2) \in F$; \label{chcon}
			\item $\forall p \in P \colon$ if
			$k(p) = c$, then ($\pre{p} \neq \varnothing$ and  $\forall t \in \pre{p} \colon h(t) = c!$) and \\ ($\post{p} \neq \varnothing$ and $\forall t \in \post{p} \colon h(t)=c?$).
		\end{enumerate}
	\end{enumerate}		
\end{definition} 

\begin{figure}[h]
	\centering
	\includegraphics[width=5.5cm]{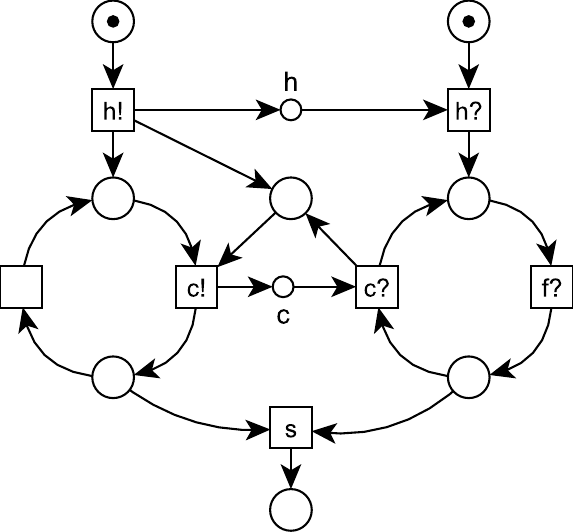}
	\caption{Labeled generalized workflow net}\label{lgwf_ex}
\end{figure}

By Definition \ref{lwf_def}, it is easy to see that there is a \emph{unique} place labeled by ``$c$'' connecting only nonempty sets of transitions with complement labels ``$c!$'' and ``$c?$'' in an LGWF-net.
This place is called a channel.
Other unlabeled places can also connect transitions with complement labels.
For instance, in Fig.~\ref{lgwf_ex}, there is the unique place labeled by ``$h$'' with the single incoming arc from transition ``$h!$'' and the single outgoing arc to transition ``$h?$''. 
However, there is no place labeled by ``$f$'' in this LGWF-net, since there are no sending transitions labeled by ``$f!$''.

We also note that the number of transitions with ``$c!$'' label is not less than the number of transitions with ``$c?$'' label in any firing sequence of an LGWF-net, for any place labeled by $c$.
In other words, the number of times one can receive a message from a channel cannot be greater than the number of times a message has been sent to this channel.
This follows from the fact that transition $c?$ can fire only after transition $c!$ if the latter is present in an LGWF-net.
There will be a unique labeled place $c$ that is an input place to transition $c?$ and an output place to transition $c!$.

Let $N^- = (P, T, F, m_0, m_f)$ denote the \emph{underlying} GWF-net obtained from an LGWF-net $N = (P, T, F, m_0, m_f, h, \ell, k)$ by removing labels from transitions and places.
Correspondingly, an LGWF-net $N$ is sound if its underlying GWF-net $N^-$ is sound.

\subsection{AS-Composition of LGWF-Nets}

Here we define an \emph{AS-composition} of LGWF-nets.
It captures synchronous and asynchronous interactions among system components according to transition labels.
The AS-composition of LGWF-nets yields a complete model of a distributed system.

The AS-composition is defined for \emph{structurally disjoint} LGWF-nets. 
Intuitively, to compose LGWF-nets, it is necessary to:
\begin{itemize}
	\item add and connect labeled places with transitions having complement asynchronous labels;
	\item merge transitions with identical synchronous labels.
\end{itemize}

The formalization of the AS-composition is given in Definition \ref{comp} for the basic case of composing two LGWF-nets.
It is easy to see that both channel addition and transition synchronization do not lead to the violation of the structural requirements imposed by Definition \ref{GWF} for a GWF-net. 
Thus, in this definition, we explicitly construct an LGWF-net by the AS-composition.

\begin{definition}\label{comp}
	Let $N_i = (P_i, T_i, F_i, m_0^i, m_f^i, h_i, \ell_i, k_i)$ be an LGWF-net for $i=1,2$, such that $(P_1 \cup T_1) \cap (P_2 \cup T_2) = \varnothing$.
	Let $P_i^u = P_i \setminus \dom{k_i}$ and $T_i^a = T_i \setminus \dom{\ell_i}$ for $i=1, 2$.
	The AS-composition of $N_1$ and $N_2$, denoted $N_1\! \circledast \!N_2$, is the LGWF-net $(P, T, F, m_0, m_f, h, \ell, k)$, where:
	\begin{enumerate}
		\item $P = P_1^u \cup P_2^u \cup P_c$ where \\
		$\abs{P_c} = \abs{C},$ 
		$C= \{ c \in \mathbf{ch}(\rng{h}) \, \vert \, \exists t, t' \in T_1^a \cup T_2^a \colon \mathbf{ch}(h(t)) = c \land  h(t) = \overline{h(t')}\}$.\label{opluspl}
		\item $m_0 = m_0^1 \cup m_0^2$ and $m_f = m_f^1 \cup m_f^2$.\label{oplusinst}
		\item $T = T_1^a \cup T_2^a \cup T_{sync}$ where \\
		$T_{sync} = \{(t_1, t_2)\, \vert \, t_1 \in \dom{\ell_1}, t_2 \in \dom{\ell_2}, \ell_1(t_1) = \ell_2(t_2)\}$.
		\item $F$ is defined by the following four cases:
		\begin{enumerate}
			\item $\forall p \in P_i^u , \forall t \in T_i^a$ for $i=1, 2$ 
			\begin{itemize}
				\item $(p, t) \in F \Leftrightarrow (p,t) \in F_i$ and
				\item $(t, p) \in F \Leftrightarrow (t,p) \in F_i$.
			\end{itemize}  
			\item $\forall p \in P_1^u, \forall t = (t_1, t_2) \in T_{sync}$
			\begin{itemize}
				\item $(p, t) \in F \Leftrightarrow (p, t_1) \in F_1$ and
				\item $(t, p) \in F\Leftrightarrow (t_1, p) \in F_1$.
			\end{itemize}
			\item $\forall p \in P_2^u, \forall t = (t_1, t_2) \in T_{sync}$
			\begin{itemize}
				\item $(p, t) \in F \Leftrightarrow (p, t_2) \in F_2$ and 
				\item $(t, p) \in F\Leftrightarrow (t_2, p) \in F_2$.
			\end{itemize}
			\item $\forall p \in P_c , \forall t \in T_i^a$ for $i=1, 2$
			\begin{itemize}
				\item $(k(p) = c) \land (h_i(t) = c!) \Rightarrow (t,p) \in F$ and
				\item $(k(p) = c) \land (h_i(t) = c?) \Rightarrow (p,t) \in F$.
			\end{itemize}
		\end{enumerate}
		\item $h \colon T \nrightarrow \Lambda$, such that $\forall t \in T_{sync} \colon h(t) = \perp $ and $\forall t \in T_i^a \colon h(t) = h_i(t)$ for  $i=1, 2$.
		\item $\ell \colon T \nrightarrow S$, such that $\forall t = (t_1, t_2) \in T_{sync} \colon \ell(t) = \ell_1(t_1) = \ell_2(t_2)$ and \\
		$\forall t_i \in T_i^a \colon \ell(t_i) = \perp$ for $i=1, 2$.
		\item $k \colon P \nrightarrow C$, such that $k\vert_{P_c}$ is a bijection and $\forall p \notin P_c \colon k(p) = \perp$.\label{opluspllab}
	\end{enumerate}
\end{definition}

Consider the example shown in Fig.~\ref{comp_ex}.
We compose two LGWF-nets $N_1$ and $N_2$ shown in Fig.\,\ref{comp_ex1}.
They exchange messages via two channels $x$ and $y$. 
They also synchronize when transitions $b$ and $f$ fire.
This fact is given by the common synchronization label $s$ of transitions $b$ and $f$.
As a result, we need to introduce two labeled places $x$ and $y$ and connect them according to the sending/receiving labels of transitions.
In addition, we merge transitions $b$, $f$ obtaining a single transition $(b,f)$ in the AS-composition $N_1\! \circledast\! N_2$ shown in Fig.\,\ref{comp_ex2}.

Note also that after synchronizing transitions, the AS-composition can have places, whose neighborhoods coincide.
Such places \emph{can} be merged into a single one to make a net \emph{P-simple}.
Correspondingly, in Fig.\,\ref{comp_ex2}, we have merged the output places of the synchronized transitions $b$ and $f$.
\begin{figure}[!h]
	\centering
	\begin{subfigure}[b]{0.5\textwidth}
		\centering
		\includegraphics[height=4.5cm]{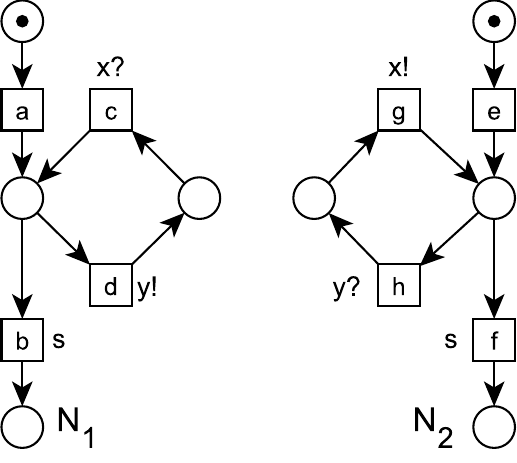}
		\caption{interacting components\label{comp_ex1}}
	\end{subfigure}%
	\begin{subfigure}[b]{0.5\textwidth}
		\centering
		\includegraphics[height=4.5cm]{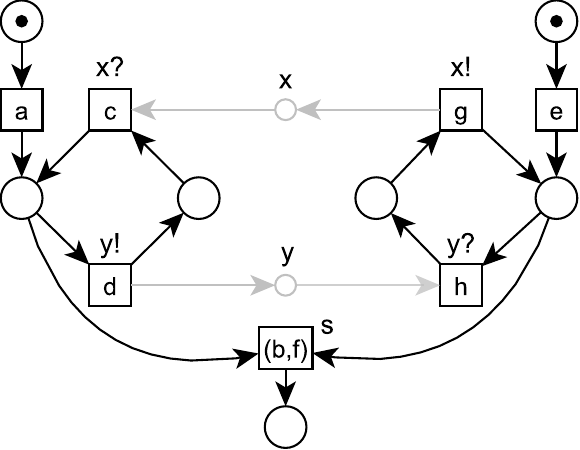}
		\caption{$N_1 \circledast N_2$\label{comp_ex2}}
	\end{subfigure}
	\caption{AS-composition of two LGWF-nets\label{comp_ex}}
\end{figure}

The AS-composition of LGWF-nets enjoys several properties that are easy to verify.
Firstly, it is both a commutative and associative operation.
These algebraic properties directly follow from the construction rules.
Thus, it can be generalized to the case of composing more than two LGWF-nets.

Secondly, a reachable marking in an AS-composition of LGWF-nets can be decomposed into three ``sub-markings'': the reachable markings of component LGWF-nets together with a marking of labeled places.
This follows from the fact that we can project the AS-composition firing sequences on the transitions in component LGWF-nets.
Then we obtain the corresponding firing sequences of components (see Proposition \ref{markdec}, where we formalize this property).

\begin{proposition}\label{markdec}
	Let $N_i = (P_i, T_i, F_i, m_0^i, m_f^i, h_i, \ell_i, k_i)$ be an LGWF-net for $i=1, 2$, and $N_1\!\circledast\!N_2 = (P, T, F, m_0, m_f, h, \ell, k)$ be the AS-composition of $N_1$ and $N_2$.
	Then $\forall m \in \reach{m_0} \colon m = (m_1 \setminus \dom{k_1}) \cup (m_2 \setminus \dom{k_2}) \cup m_c$, where $m_1 \in \reach{m_0^1}$, $m_2 \in \reach{m_0^2}$ and $m_c \subseteq \dom{k}$.
\end{proposition}   

However, the AS-composition of LGWF-nets studied above may not preserve behavioral and structural properties of component LGWF-nets.
For instance, if $N_1$ and $N_2$ are two sound LGWF-nets, their composition $N_1\! \circledast\! N_2$ might not be sound.
Consider the example provided in Fig.\,\ref{nsound}, where the system $N_1\! \circledast\! N_2$ is composed of two sound LGWF-nets.
$N_1\! \circledast\! N_2$ may reach a final marking $\{f_1, s_2\}$ different from the expected final marking $\{f_1, f_2\}$ if $N_1$ does not send a message to channel $d$.
This makes $N_1\! \circledast\! N_2$ lose soundness. 
Moreover, it is no longer covered by sequential components.
\begin{figure}[h]
	\centering
	\includegraphics[height=4.5cm]{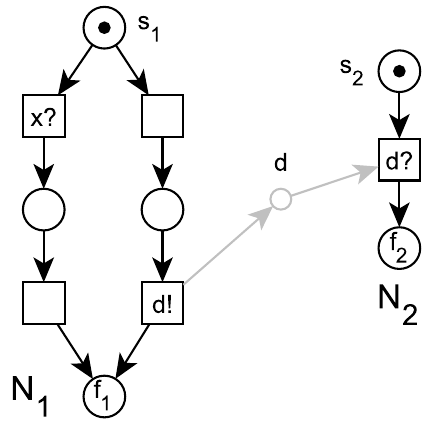}
	\caption{AS-composition may not preserve component properties}\label{nsound}
\end{figure}

The preservation of component properties in their AS-composition is the main problem we address in the paper.
For this purpose, instead of considering the AS-composition of LGWF-nets directly, we will analyze an underlying \emph{abstract interface}.
It is an LGWF-net that models how components interact.
Component models are mapped on corresponding subnets in an interface net via \emph{morphisms} discussed in the following section.
Then in Section \ref{sec:main}, we will apply these morphisms to achieve the preservation of component soundness in the AS-composition.

\section{Abstraction and Refinement in GWF-Nets Based on Morphisms}\label{sec:alpha}
This section describes a basic technique supporting abstraction and refinement in Petri nets based on $\alpha$-morphisms.
We study the properties of $\alpha$-morphisms relevant to GWF-nets.
These properties are further used to address the problem of preserving the soundness of LGWF-nets in their AS-composition.

\subsection{Place Refinement and $\alpha$-Morphisms}
The class of \emph{$\alpha$-morphisms} was introduced in \cite{Bernardinello2013} to support abstraction and refinement in Petri nets covered by sequential components.
An $\alpha$-morphism example is provided in Fig. \ref{alpha_ex}, where refinement of places is depicted by shaded subnets, i.e., the subnet $N_1(\omr{p_2})$ in $N_1$ refines the place $p_2$ in $N_2$. 
Refinement of transitions is explicitly given by their names, i.e., two transitions $f_1$ and $f_2$ in $N_1$ refine the same transition $f$ in $N_2$.
In other words, refinement may lead to splitting transitions of an abstract net.
After providing the formal definition of $\alpha$-morphisms, we also discuss the general intuition behind them.

\begin{definition}\label{alpham}
	Let $N_i = (P_i, T_i, F_i, m_0^i)$ be an SMD Petri net, $X_i = P_i \cup T_i$ for $i=1, 2$, where $X_1 \cap X_2 = \varnothing$. 
	An $\alpha$-morphism from $N_1$ to $N_2$ is a total surjective map $\varphi \colon X_1 \to X_2$, also denoted $\varphi \colon N_1 \to N_2$, where:
	\begin{enumerate}
		\item $\varphi(P_1) = P_2$.
		\item $\varphi(m_0^1) = m_0^2$.\label{instate}
		\item $\forall t_1 \in T_1 \colon$ if $\om{t_1} \in T_2$, then $\om{\pre{t_1}}=\pre{\om{t_1}}$ and $\om{\post{t_1}}=\om{t_1}^\bullet$.\label{tTOt}
		\item $\forall t_1 \in T_1 \colon$ if $\om{t_1} \in P_2$, then $\om{\neighb{t_1}}=\{\om{t_1}\}$.\label{tTOp}
		\item $\forall p_2 \in P_2 \colon$
		\begin{enumerate}
			\item $N_1(\varphi^{-1}(p_2))$ is an acyclic net or $\omr{p_2} \subseteq P_1$.\label{acycsub}
			\item $\forall p_1 \in \inp{N_1(\varphi^{-1}(p_2))} \colon \om{\pre{p_1}} \subseteq \pre{p_2}$ and if $\pre{p_2} \neq \varnothing$, then $\pre{p_1} \neq \varnothing$.\label{inP}
			\item $\forall p_1 \in \outp{N_1(\varphi^{-1}(p_2))} \colon \om{\post{p_1}} = \post{p_2}$.\label{outP}
			\item $\forall p_1 \in P_1 \cap \varphi^{-1}(p_2) \colon p_1 \notin \inp{N_1(\varphi^{-1}(p_2))} \Rightarrow \om{\pre{p_1}}=p_2 \text{ and }$\\ $p \notin \outp{N_1(\varphi^{-1}(p_2))} \Rightarrow \om{\post{p_1}} = p_2$.\label{INsub}
			\item $\forall p_1 \in P_1 \cap \varphi^{-1}(p_2)\colon$ there is a sequential component $N' = (P', T', F')$ in $N_1$, such that $p_1 \in P'$, $\varphi^{-1}(\neighb{p_2}) \subseteq T'$.\label{scomp}
		\end{enumerate}
	\end{enumerate}
\end{definition}

\begin{figure}[h]
	\centering
	\includegraphics[height=6cm]{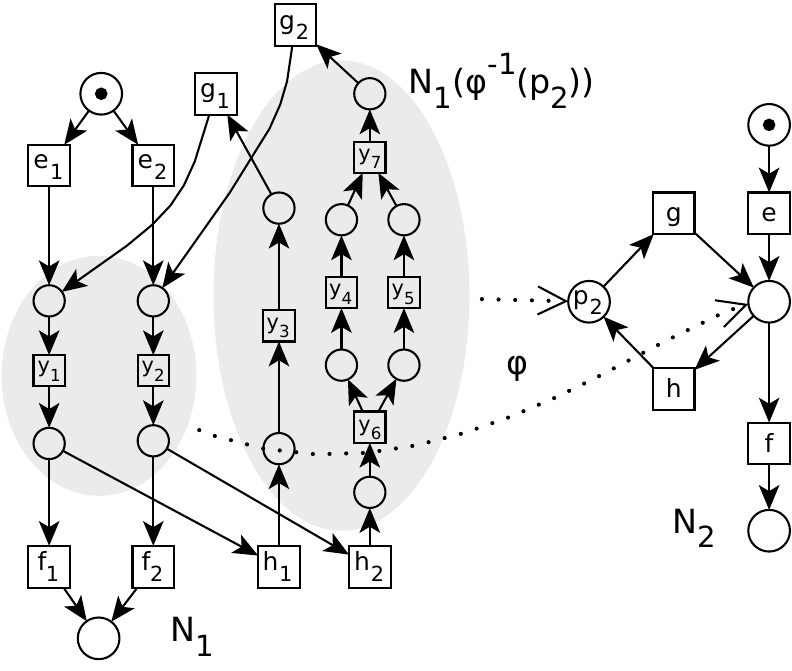}
	\caption{The $\alpha$-morphism $\varphi \colon N_1 \to N_2$}\label{alpha_ex}
\end{figure}

According to the definition, $\alpha$-morphisms allow us to refine places in $N_2$ by replacing them with \emph{acyclic} subnets in $N_1$, where $N_2$ is called an abstract net and $N_1$ is its refinement. 
If a transition in $N_1$ is mapped to a transition in $N_2$, then the neighborhood of the transition in $N_1$ should be mapped on the neighborhood of the corresponding transition in $N_2$
(by Definition \ref{alpham}.\ref{tTOt}). 
If a transition in $N_1$ is mapped to a place in $N_2$, then the neighborhood of this transition should be mapped to the same place (by Definition \ref{alpham}.\ref{tTOp}).

The main motivation behind $\alpha$-morphisms is the possibility to ensure that the behavioral properties of an abstract model hold in its refinement as well.
Therefore, each output place in a subnet should have the same choices as its abstraction does (by Definition \ref{alpham}.\ref{outP}).
Input places do not need this constraint (by Definition \ref{alpham}.\ref{inP}).
A choice between them is made before, since there are no concurrent transitions in the neighborhood of a subnet (by Definition \ref{alpham}.\ref{scomp}).
Moreover, by Definition \ref{alpham}.\ref{INsub}, neighborhoods of places internal to a subnet are mapped to the same place as the subnet. 

To sum up, requirements imposed by Definition \ref{alpham}.\ref{acycsub}-\ref{scomp} ensure the main intuition behind $\alpha$-morphisms. 
If a subnet in $N_1$ refines a place in $N_2$, then this subnet should behave ``in the same way'' as the place in $N_2$ does.
More precisely, let $N_1(\omr{p_2})$ be a subnet in $N_1$ refining a place $p_2$ in $N_2$. Then the following holds:
\begin{enumerate}
	\item No tokens are left in $N_1(\omr{p_2}) \cap P_1$ after firing an output transition in $\post{(\outp{N_1(\omr{p_2})})}$;
	\item No transitions are enabled in $\pre{(\inp{N_1(\omr{p_2})})}$ whenever there is a token in $N_1(\omr{p_2}) \cap P_1$.
\end{enumerate}

\subsection{Properties Preserved and Reflected by $\alpha$-Morphisms}\label{ssec:prop}
Here we study properties \emph{preserved} and \emph{reflected} by $\alpha$-morphisms (see Fig.~\ref{propscheme}).
In \cite{Bernardinello2013} several properties of $\alpha$-morphisms have already been studied.
We will refer to some of the proven properties here and consider other properties of $\alpha$-morphisms relevant to generalized workflow nets.

\begin{figure}[h]
	\centering
	\includegraphics[height=3cm]{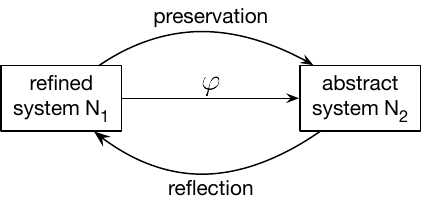}
	\caption{Relation between properties of an abstract system and its refinement}\label{propscheme}
\end{figure} 

The following proposition states that the structure of GWF-nets is preserved by $\alpha$-morphisms.

\begin{proposition}\label{strpr}
	Let $N_i=(P_i, T_i, F_i,$ $m_0^i)$ be an SMD Petri net, and $X_i = P_i \cup T_i$ for $i=1, 2$, such that there is an $\alpha$-morphism $\varphi \colon N_1 \to N_2$.
	If $N_1$ is a GWF-net, then $N_2$ is a GWF-net.
\end{proposition} 
\begin{proof}
	We show that $N_2$ satisfies the three structural conditions of GWF-nets, see Definition \ref{GWF}.
	
	\textbf{1.} By Definition \ref{alpham}.\ref{instate}, $\om{m_0^1} = m_0^2$.
	Suppose $\exists p_2 \in m_0^2 \colon \pre{p_2} \neq \varnothing$.
	By Definition \ref{alpham}.\ref{inP}, $\forall p_1 \in \inp{N_1(\varphi^{-1}(p_2))} \colon$ if $\pre{p_2} \neq \varnothing$, then $\pre{p_1} \neq \varnothing$.
	Take $p_1 \in m_0^1$, such that $\om{p_1} = p_2$. 
	Since $p_1 \in \inp{N_1(\varphi^{-1}(p_2))}$, then $\pre{p_1} \neq \varnothing$.
	By Definition \ref{GWF}.\ref{instwf}, $\forall p \in m_0^1 \colon \pre{p} = \varnothing$.
	Then, $\pre{p_2} = \varnothing$ and $\forall p \in m_0^2 \colon \pre{p} = \varnothing$.
	
	\textbf{2.} By Definition \ref{GWF}.\ref{outstwf}, $m_f^1 \subseteq P_1$, such that $\post{(m_f^1)} = \varnothing$.
	Denote $\om{m_f^1}$ by $m_f^2\!\subseteq\!P_2$.
	Suppose $\exists p_2 \in m_f^2 \colon \post{p_2} \neq \varnothing$.
	Take $p_1 \in m_f^1$, such that $\om{p_1} = p_2$.
	By Definition \ref{alpham}.\ref{outP}, $\forall p_1 \in \outp{N_1(\varphi^{-1}(p_2))} \colon \om{\post{p_1}} = \post{p_2}$.
	Since $p_1 \in \outp{N_1(\varphi^{-1}(p_2))}$, $\post{p_1} \neq \varnothing$.
	But by Definition \ref{GWF}.\ref{outstwf}, $p_1 \in m_f^1$ and $\post{p_1} = \varnothing$.
	Then, $\post{p_2} = \varnothing$ and $\forall p \in m_f^2 \colon \post{p} = \varnothing$.
	
	\textbf{3.} Suppose $\exists x_2 \in X_2$, such that $\forall p \in m_0^2 \colon (p, x_2) \notin F_2^*$.
	Since an $\alpha$-morphism is a surjective map, $\omr{x_2} \neq \varnothing$.
	Thus, $\varphi^{-1}(x_2) = \{x_1^1, \dots, x_1^k\} \subseteq X_1$, where $k \geq 1$.
	If $x_2 \in T_2$, then $\varphi^{-1}(x_2) \subseteq T_1$, and we take $x_1 \in \omr{x_2}$.
	If $x_2 \in P_2$, then we take $x_1 \in \inp{N_1(\omr{x_2})}$.
	By Definition \ref{GWF}.\ref{connwf}, $\exists s \in m_0^1 \colon (s,x_1) \in F_1^*$.
	Then, $\om{\pre{x_1}} \in \pre{x_2}$ or $\om{\pre{x_1}} = x_2$.
	We follow backward the whole path from $s$ to $x_1$ in $N_1$ mapping it on $N_2$ with $\varphi$.
	Thus, we obtain that $\exists x' \in X_2 \colon (x', x_2) \in F_2^*$ and $\om{s}=x'$, where $x' \in m_0^2$.
	
	Suppose $\exists x_2 \in X_2$, such that $\forall p \in m_f^2 \colon (x_2, p) \notin F_2^*$.
	Since an $\alpha$-morphism is a surjective map, $\omr{x_2} \neq \varnothing$.
	Thus, $\varphi^{-1}(x_2) = \{x_1^1, \dots, x_1^k\} \subseteq X_1$, where $k \geq 1$.
	If $x_2 \in T_2$, then $\varphi^{-1}(x_2) \subseteq T_1$, and we take $x_1 \in \omr{x_2}$.
	If $x_2 \in P_2$, then we take $x_1 \in \outp{N_1(\omr{x_2})}$.
	By Definition \ref{GWF}.\ref{connwf}, $\exists f \in m_f^1 \colon (x_1, f) \in F_1^*$.
	Then, $\om{\post{x_1}} \in \post{x_2}$ or $\om{\post{x_1}} = x_2$.
	We follow the whole path forward from $x_1$ to $f$ in $N_1$ mapping it on $N_2$ with $\varphi$.
	Thus, we obtain that $\exists x' \in X_2 \colon (x_2, x') \in F_2^*$ and $\om{f} = x'$, where $x' \in m_f^2$. \qed
\end{proof}

It follows from Proposition \ref{strpr} that $\om{m_f^1} = m_f^2$, i.e., final markings of GWF-nets are also preserved by $\alpha$-morphisms.
In the general case, the converse of Proposition \ref{strpr} is not valid. 
Indeed, $\alpha$-morphisms do not reflect the initial markings of GWF-nets properly (see Fig.\,\ref{gwf_nonrefl}).

A refined net $N_1$ is \emph{well marked} w.r.t. $\varphi$ if each input place in a subnet in $N_1$, refining a marked place in an abstract net $N_2$, is marked as well.
Consider again the $\alpha$-morphism shown in Fig. \ref{gwf_nonrefl}, the token in the shaded subnet must be placed into $p$ to make $N_1$ well marked w.r.t. to $\varphi$.

In the following proposition, we prove that $\alpha$-morphisms reflect the structure of GWF-nets under the well-markedness of $N_1$.

\begin{proposition}\label{strref}
	Let $N_i=(P_i, T_i, F_i,$ $m_0^i)$ be an SMD Petri net, and $X_i = P_i \cup T_i$ for $i=1, 2$, such that there is an $\alpha$-morphism $\varphi \colon N_1 \to N_2$.
	If $N_2$ is a GWF-net and $N_1$ is well marked w.r.t. $\varphi$, then $N_1$ is a GWF-net.
\end{proposition}
\begin{proof}
	We show that $N_1$ satisfies the three structural conditions of GWF-nets, see Definition \ref{GWF}.
	
	\textbf{1.} By Definition \ref{GWF}.\ref{instwf}, $\forall s_2 \in m_0^2 \colon \pre{s_2} = \varnothing$.
	Since $N_1$ is well-marked w.r.t. $\varphi$, $m_0^1 = \{ \inp{N_1(\omr{s_2})} \,\vert\, s_2 \in m_0^2 \}$.
	Take $s_2 \in m_0^2$ and the corresponding subnet $N_1(\omr{s_2})$.
	Suppose $\exists p \in \inp{N_1(\omr{s_2})}$, such that $\pre{p} \neq \varnothing$.
	Then $\om{p} = s_2$ and, by Definition \ref{alpham}.\ref{inP}, $\om{\pre{p}} \subseteq \pre{s_2} = \varnothing$ contradicting the total surjectivity of $\varphi$.
	
	\textbf{2.} By Definition \ref{GWF}.\ref{outstwf}, $\forall f_2 \in m_f^2 \colon \post{f_2} = \varnothing$.
	Take $f_2 \in m_f^2$ and the corresponding subnet $N_1(\omr{f_2})$.
	Also take $p \in \outp{N_1(\omr{f_2})}$.
	Then $\om{p} = f_2$.
	By Definition~\ref{alpham}.\ref{outP}, $\om{\post{p}} = \post{f_2} = \varnothing$ contradicting the total surjectivity of $\varphi$.
	Thus, we obtain the final marking of $N_1$, i.e., $m_f^1 = \{ \outp{N_1(\omr{f_2})} \,\vert\, f_2 \in m_f^2 \}$ and $\post{(m_f^1)} = \varnothing$.
	
	\textbf{3.} Suppose $\exists x_1 \in X_1$, such that $\forall s_1 \in m_0^1 \colon (s_1, x_1) \notin F_1^*$.
	If $(x_1, x_1) \notin F_1^*$, we follow the path from $x_1$ to the first node $x_1' \in X_1$ in $N_1$ backward, such that $\pre{x_1'} = \varnothing$.
	Since $\forall t_1 \in T_1 \colon \abs{\pre{t_1}} \geq 1$, $x_1' \in P_1$.
	If $x_1' \notin m_0^1$, then $N_1$ is not well-marked w.r.t. $\varphi$. 
	If $(x_1, x_1) \in F_1^*$, then, by Definition \ref{alpham}.\ref{acycsub}, there is a corresponding image cycle in $N_2$.
	Take $x_2 \in X_2$, such that $\om{x_1} = x_2$.
	By Definition \ref{GWF}.\ref{outstwf}, $\exists s_2 \in m_0^2 \colon (s_2, x_2) \in F_2^*$.
	Take $x_2' \in X_2$ belonging to this cycle, where at least one node in $\pre{x_2'}$ is not in the cycle.
	By surjectivity of $\varphi$, $\exists x_1' \in X_1 \colon \om{x_1'} = x_2'$ belonging to the cycle $(x_1, x_1) \in F_1^*$.
	If $x_2' \in T_2$, then $\omr{x_2'} \subseteq T_1$.
	By Definition \ref{alpham}.\ref{tTOt}, the neighborhood of transitions is preserved by $\varphi$.
	Then, $\forall t_1 \in \omr{x_2'} \colon \om{\pre{t_1}} = \pre{x_2'}$, i.e., there is a place in $\pre{\omr{x_2'}}$ which does not belong to the cycle $(x_1, x_1) \in F_1^*$.
	If $x_2' \in P_2$, then take $\inp{N_1(\omr{x_2'})}$.
	At least one place in $\inp{N_1(\omr{x_2'})}$ has an input transition which does not belong to the cycle $(x_1, x_1) \in F_1^*$, since there is a node in $\pre{x_2'}$ which is not in the image cycle in $N_2$.
	We have shown that $\exists x \in \pre{x_1'}$, such that $x$ does not belong to the cycle $(x_1, x_1) \in F_1^*$.
	Thus, either there is a path from $\widetilde{x}$ to $x$ with $\pre{\widetilde{x}} = \varnothing$, or there is another cycle $(\widetilde{x}, \widetilde{x}) \in F_1^*$.
	
	Applying a similar reasoning, we prove that $\forall x_1 \in X_1 \, \exists f_1 \in m_f^1 \colon (x_1, f_1) \in F_1^*$.
	The only difference is that we follow paths forward. \qed
\end{proof}

\begin{figure}[h]
	\centering
	\begin{subfigure}[b]{0.5\textwidth}
		\centering
		\includegraphics[height=4.5cm]{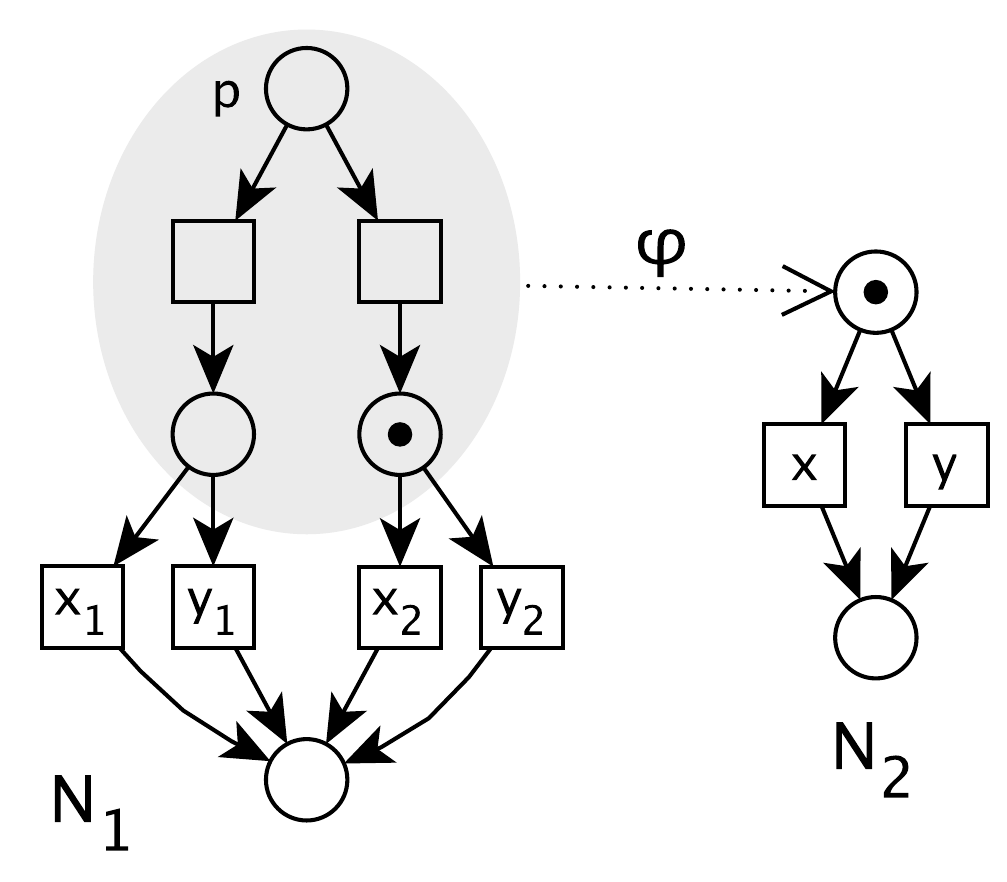}
		\caption{\label{gwf_nonrefl}}
	\end{subfigure}%
	\begin{subfigure}[b]{0.5\textwidth}
		\centering
		\includegraphics[height=4.5cm]{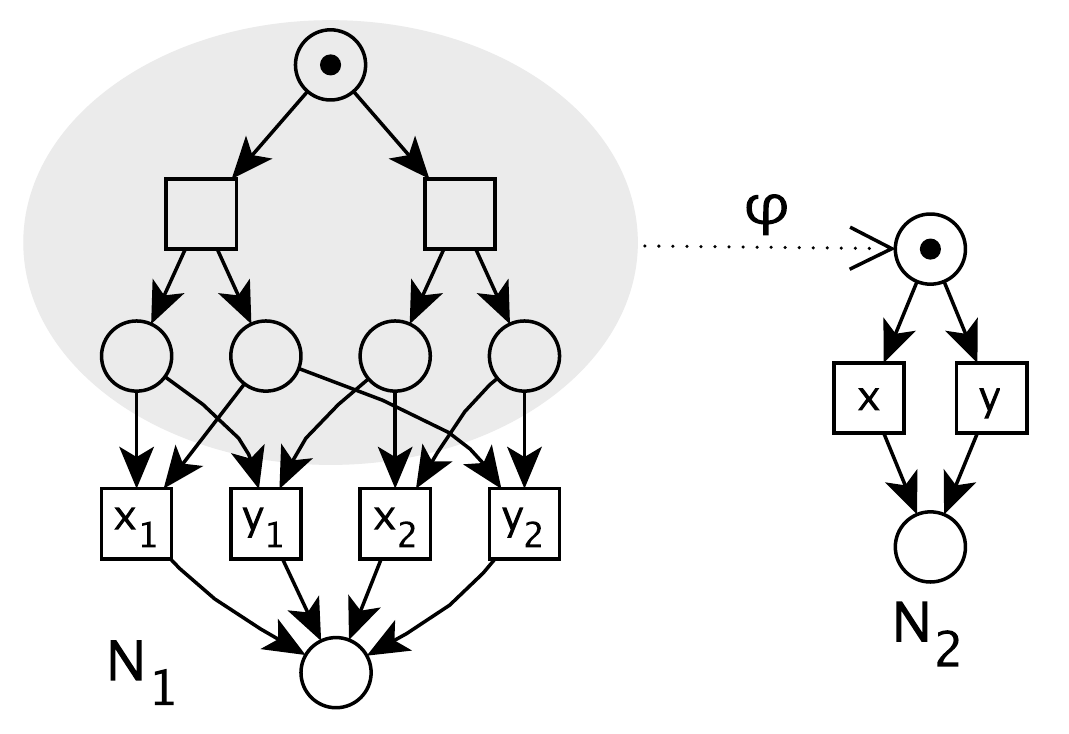}
		\caption{\label{sound_nonrefl}}
	\end{subfigure}
	\caption{Two $\alpha$-morphisms with the same range}
\end{figure}

In Propositions \ref{strpr} and \ref{strref}, we have proven two structural properties of $\alpha$-morphisms.
Further, we study preservation and reflection of behavioral properties, i.e., whether reachable markings are preserved and reflected by $\alpha$-morphisms.

According to the following proposition, it was proven that $\alpha$-morphisms preserve reachable markings and transition firings.

\begin{proposition}[see \cite{Bernardinello2013}]
	Let $N_i=(P_i, T_i, F_i,$ $m_0^i)$ be an SMD Petri net, and $X_i = P_i \cup T_i$ for $i=1, 2$, such that there is an $\alpha$-morphism $\varphi \colon N_1 \to N_2$.
	Let $m_1 \in \reach{m_0^1}$. Then $\om{m_1} \in \reach{m_0^2}$. 
	If $m_1 \reach{t}m_1'$, where $t \in T_1$, then:
	\begin{enumerate}
		\item $\om{t} \in T_2 \Rightarrow \om{m_1}\reach{\om{t}}\om{m_1'}$.
		\item $\om{t} \in P_2 \Rightarrow \om{m_1} = \om{m_1'}$.
	\end{enumerate}\label{markpres}
\end{proposition}

In the general case, $\alpha$-morphisms do not reflect both reachable markings and transition firings.
More precisely, having a reachable marking $m_2 \in \reach{m_0^2}$, such that $m_2\reach{t_2}$ for a transition $t_2$ in $N_2$, we cannot say that $\forall t_1 \in \omr{t_2}\,  \exists m_1 = \omr{m_2} \in \reach{m_0^1} \colon m_1 \reach{t_1}$ in $N_1$.

Note that the reflection of reachable markings is a crucial property, since we seek to deduce the behavioral properties of a refined system from those of its abstraction.
It is necessary to check additional \emph{local} conditions based on the unfolding to achieve the reflection of reachable markings.
We briefly describe this technique first introduced in \cite{Bernardinello2013}.

Let $N_1$ and $N_2$ be two Petri nets related via the $\alpha$-morphism $\varphi: N_1 \to N_2$, 
Recall that $N_1$ is called a refinement, and $N_2$ is called an abstraction of $N_1$.
For every place $p_2$ in $N_2$, refined by a subnet in $N_1$, we construct a \emph{local} net, denoted by $S_2(p_2)$, by taking the neighborhood transitions of $p_2$ with artificial input and output places if necessary.
The same is done for the refined system $N_1$. 
We construct the corresponding local net, denoted by $S_1(p_2)$, by taking the subnet in $N_1$ refining $p_2$ via $\varphi$, i.e., $N_1(\omr{p_2})$ and the transitions $\varphi^{-1} (\pre{p_2}) \cup \varphi^{-1} (\post{p_2})$ with artificial input and output places if necessary.
As a result, we have two local nets $S_1(p_2)$ and $S_2(p_2)$.

Since there is the $\alpha$-morphism $\varphi: N_1 \to N_2$, there is also the $\alpha$-morphism $\varphi^S \colon S_1(p_2) \to S_2(p_2)$ corresponding to the restriction of $\varphi$ to the places and transitions in $S_1(p_2)$.
Recall that the unfolding of a Petri net $N$, denoted by $\mathcal{U}(N)$, is the maximal branching process of $N$, such that any other branching process is isomorphic to a subnet in $\mathcal{U}(N)$. The nodes in $\mathcal{U}(N)$ are mapped to the nodes in $N$ via the \emph{folding} function $u$.
In Lemma \ref{unf}, taking the unfolding of $S_1(p_2)$, we prove that the associated folding function $u$ composed with the $\alpha$-morphisms $\varphi^S$ is also the $\alpha$-morphism under the soundness of $N_1$.
Note that since $S_1(p_2)$ is acyclic (by Definition \ref{alpham}.\ref{acycsub}), its unfolding is finite.
This helps us to assure that the ``final'' marking in a subnet in $N_1$, refining a place $p_2$ in the abstract model $N_2$, enables exactly the inverse image of transitions in $\post{p_2}$.
After providing Lemma \ref{unf}, we also discuss a specific example of checking these local conditions.

\begin{lemma}\label{unf}
	Let $N_i=(P_i, T_i, F_i,$ $m_0^i)$ be an SMD Petri net, and $X_i = P_i \cup T_i$ for $i=1, 2$, such that there is an $\alpha$-morphism $\varphi \colon N_1 \to N_2$.
	Let $\mathcal{U}(S_1(p_2))$ be the unfolding of $S_1(p_2)$ with the folding function $u$, and $\varphi^S$ be an $\alpha$-morphism from $S_1(p_2)$ to $S_2(p_2)$, where $p_2 \in P_2$.
	Let $N_1$ be a sound GWF-net. Then, the map from $\mathcal{U}(S_1(p_2))$ to $S_2(p_2)$ obtained as $\varphi^S \circ u$ is an $\alpha$-morphism.
\end{lemma}
\begin{proof}
	Since $N_1$  is a GWF-net, $S_1(p_2)$ is also a GWF-net.
	By Lemma 1 of \cite{Bernardinello2013}, when a transition in $\omr{\post{p_2}}$ fires, it empties the subnet $N_1(\omr{p_2})$.
	Then $S_1(p_2)$ is sound, and, by Def.\ref{GWF}.\ref{sndlive}, each transition in $S_1(p_2)$ will occur at least once.
	Thus, the folding $u$ is a surjective function from $\unf{S_1(p_2)}$ to $S_1(p_2)$ and the composition $\varphi^S \circ u$ is the $\alpha$-morphism from $\unf{S_1(p_2)}$ to $S_2(p_2)$. \qed
\end{proof}

Figure \ref{locunf} provides a negative example of checking the local unfolding condition when $N_1$ is not sound.
We use the $\alpha$-morphism previously shown in Fig.\,\ref{sound_nonrefl}.
In this case, local nets coincides with the original $N_1$ and $N_2$.
When we unfold $N_1$, there are no occurrences of transitions $y_1$ and $y_2$. 
Thus, the composition of the corresponding folding function and the $\alpha$-morphism $\varphi \circ u$ is not an $\alpha$-morphism.
The ``final'' marking in the subnet $N_1(\omr{p_2})$, which refines $p_2$ in $N_2$, enables transitions $x_1$ and $x_2$ only, whereas, in $N_2$, transition $y$ is also enabled.
Therefore, transitions in the inverse image of $y$ in $N_2$ cannot be enabled by the final marking in the subnet $N_1(\omr{p_2})$.
\begin{figure}[h]
	\centering
	\includegraphics[height=4.5cm]{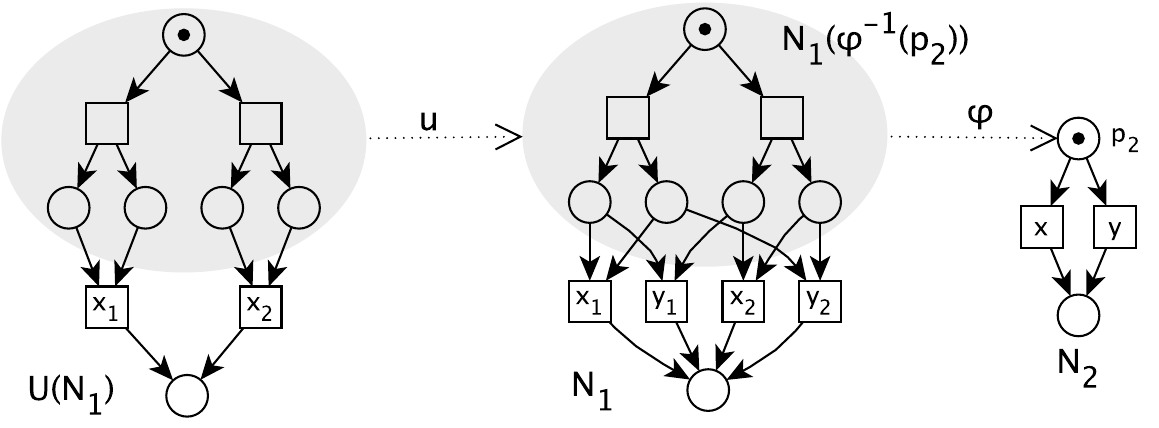}
	\caption{A negative example of checking local conditions based on the unfolding}\label{locunf}
\end{figure}

One should check the unfolding condition for all properly refined places in an abstract net.
A properly refined place is a place that is refined by a subnet rather than by a set of places.
Taking the above discussion into account, we obtain that $\alpha$-morphisms reflect reachable markings and transition firings under the soundness of a refinement (see Proposition \ref{markrefl}).

\begin{proposition}\label{markrefl}
	Let $N_i=(P_i, T_i, F_i,$ $m_0^i)$ be an SMD Petri net, and $X_i = P_i \cup T_i$ for $i=1, 2$, such that there is an $\alpha$-morphism $\varphi \colon N_1 \to N_2$.
	Let $N_1$ be a sound GWF-net.
	Then $\forall m_2 \in \reach{m_0^2} \, \exists m_1 \in \reach{m_0^1} \colon$ $ \varphi(m_1) = m_2$.
	If $m_2\reach{t_2}$, then $\forall t \in \varphi^{-1} (t_2) \, \exists m_1 = \varphi^{-1}(m_2) \in \reach{m_0^2} \colon $ $m_1\reach{t_1}$.
\end{proposition}

\begin{proof}
	Follows from Lemma \ref{unf}. \qed
\end{proof}

The following theorem expresses the main result of this section. 
We prove that $\alpha$-morphisms preserve soundness of GWF-nets.

\begin{theorem}\label{spres}
	Let $N_i=(P_i, T_i, F_i,$ $m_0^i)$ be an SMD Petri net, and $X_i = P_i \cup T_i$ for $i=1, 2$, such that there is an $\alpha$-morphism $\varphi \colon N_1 \to N_2$.
	If $N_1$ is a sound GWF-net, then $N_2$ is a sound GWF-net.
\end{theorem}
\begin{proof}
	We show that $N_2$ satisfies the three behavioral conditions of a sound GWF-net, see Definition \ref{sound}.
	
	\textbf{1.} By Definition \ref{sound}.\ref{sndprop}, for all $m_1 \in \reach{m_0^1} \colon m_f^1 \in \reach{m_1}$.
	Then, $\exists w \in FS(N_1) \colon$ $m_1\reach{w}m_f^1$. i.e $w=t_1t_2\dots t_n$ and $m_1\reach{t_1}m_1^1\dots m_1^{n-1}\reach{t_n}m_f^1$.
	Using Proposition \ref{markpres}, it is possible to simulate $w$ in $N_2$. 
	By Proposition \ref{strpr}, $\om{m_f^1} = m_f^2$.
	Suppose $\exists m_2 \in \reach{m_0^2} \colon m_f^2 \notin \reach{m_2}$.
	By Proposition \ref{markrefl}, $\exists m_1' \in \reach{m_0^1} \colon \omr{m_2} = m_1'$.
	By Definition \ref{sound}.\ref{sndprop}, $m_f^1 \in \reach{m_1'}$.
	Thus, $\exists w' \in FS(N_1) \colon$ $m_1'\reach{w'}m_f^1$.
	Using Proposition \ref{markpres}, it is again possible to simulate $w'$ in $N_2$.
	Then, $m_f^2 \in \reach{m_2}$.
	
	\textbf{2.} Suppose $\exists m_2' \in \reach{m_0^2} \colon m_2' \supseteq m_f^2$.
	Then $m_2' = m_f^2 \cup P_2'$, where $P_2' \cap m_f^2 = \varnothing$.
	By Proposition \ref{markrefl}, take $m_1 \in \reach{m_0^1}$, such that $\omr{m_1} = m_2'$ and $m_f^1 \nsubseteq m_1$.
	By Definition \ref{sound}.\ref{sndprop}, $m_f^1 \in \reach{m_1}$ and $\exists w \in FS(N_1) \colon m_1 \reach{w} m_f^1$.
	Using Proposition \ref{markpres}, it is possible to simulate $w$ in $N_2$.
	By Proposition \ref{strpr}, $\om{m_f^1} = m_f^2$.
	The only way to completely empty places in $P_2'$ is to consume at least one token from $m_f^2$.
	Then, $\exists f_2 \in m_f^2 \colon \post{f_2} \neq \varnothing$ which contradicts Definition \ref{GWF}.\ref{outstwf}.
	
	\textbf{3.} By Definition \ref{sound}.\ref{sndlive}, $\forall t_1 \in T_1 \, \exists m_1 \in \reach{m_0^1} \colon m_1 \reach{t_1}$.
	Sincs $\varphi$ is a surjective map, $\forall t_2 \in T_2 \, \exists t_1 \in T_1 \colon \om{t_1} = t_2$.
	By Proposition \ref{markpres}, $m_1 \reach{t_1} m_1' \Rightarrow \om{m_1}\reach{\om{t_1}}\om{m_1'}$.
	Then, $\forall t_2 \in T_2 \, \exists m_2 \in \reach{m_0^2} \colon m_2 \reach{t_2}$. \qed
\end{proof}

However, the converse of Theorem \ref{spres} is not true in general. 
Consider again the example shown in Fig.\,\ref{sound_nonrefl}, where $N_2$ is sound and $N_1$ is not sound, since transitions $y_1$ and $y_2$ cannot fire.
Thus, $\alpha$-morphisms do not reflect soundness, which follows from the fact that reachable markings are not reflected in the general case.

In our study, soundness reflection is a sought property of $\alpha$-morphisms.
We apply $\alpha$-morphisms to achieve the preservation of LGWF-net soundness in their AS-composition.
Component interactions are encoded in an \emph{abstract interface}, which also represents the abstract view of a complete system.
In the next section, we provide a technique when the soundness of an abstract interface implies the soundness of a refined system model, i.e., the associated $\alpha$-morphism reflects the soundness of an abstract interface.

\section{Preserving Soundness in the AS-Composition via Morphisms}\label{sec:main}
The AS-composition of LGWF-nets, discussed in Section \ref{sec:comp}, preserves the soundness of LGWF-nets through the use of an \emph{abstract interface}.
This model provides minimal detail on the local behavior of communicating components, focusing on their synchronous and asynchronous interactions.
An abstract interface is also referred to as an \emph{interface pattern} that is the AS-composition of corresponding abstract LGWF-nets.
Abstraction of components is implemented using $\alpha$-morphisms, discussed in the previous section, that we adjust to LGWF-nets.
We aim to deduce the soundness of a refined system model by verifying the soundness of an underlying interface pattern.

If $N_1$ and $N_2$ are two LGWF-nets, then an $\alpha$-morphism $\varphi \colon N_1 \to N_2$ should additionally respect transition labeling, i.e., a labeled transition in $N_1$ can only be mapped to a transition in $N_2$ with the same label.
Then a labeled transition in $N_1$ cannot be mapped to a place in $N_2$.
We formalize these restrictions on labeled transition mapping in the following definition.

\begin{definition}\label{alphahm}
	Let $N_i = (P_i, T_i, F_i, m_0^i, m_f^i, h_i, \ell_i, k_i)$ be an SMD LGWF-net, and $X_i = P_i \cup T_i$ for $i=1, 2$.
	An $\widehat{\alpha}$-morphism from $N_1$ to $N_2$ is a total surjective map $\varphi \colon X_1 \to X_2$, also denoted $\varphi \colon N_1 \to N_2$, where:
	\begin{enumerate}
		\item $\varphi(P_1) = P_2$, such that $\forall p_1 \in \dom{k_1} \colon k_2(\varphi(p_1)) = k_1(p_1)$.
		\item $\om{m_0^1} = m_0^2$.
		\item[2'.] $\om{m_f^1} = m_f^2$.
		\item $\forall t_1 \in T_1 \colon$ if $\om{t_1} \in T_2$, then $\om{\pre{t_1}} = \pre{\om{t_1}}$ and $\om{\post{t_1}} =\om{t_1}^\bullet$.
		\item[3'.] $\forall t_1 \in \dom{h_1} \cup \dom{\ell_1} \colon \om{t_1} \in T_2$ and
		\begin{enumerate}
			\item if $t_1 \in \dom{h_1}$, then $h_2(\om{t_1}) = h_1(t_1)$;
			\item if $t_2 \in \dom{\ell_1}$, then $\ell_2(\om{t_1}) = \ell_1(t_1)$.
		\end{enumerate}
		\item $\forall t_1 \in T_1 \colon$ if $\om{t_1} \in P_2$, then $\om{\neighb{t_1}}=\{\om{t_1}\}$.\label{tTOp}
		\item $\forall p_2 \in P_2 \colon$
		\begin{enumerate}
			\item $N_1(\varphi^{-1}(p_2))$ is an acyclic net or $\omr{p_2} \subseteq P_1$.\label{acycsub}
			\item $\forall p_1 \in \inp{N_1(\varphi^{-1}(p_2))} \colon \om{\pre{p_1}} \subseteq \pre{p_2}$ and if $\pre{p_2} \neq \varnothing$, then $\pre{p_1} \neq \varnothing$.\label{inP}
			\item $\forall p_1 \in \outp{N_1(\varphi^{-1}(p_2))} \colon \om{\post{p_1}} = \post{p_2}$.\label{outP}
			\item $\forall p_1 \in P_1 \cap \varphi^{-1}(p_2) \colon p_1 \notin \inp{N_1(\varphi^{-1}(p_2))} \Rightarrow \om{\pre{p_1}}=p_2 \text{ and }$\\ $p \notin \outp{N_1(\varphi^{-1}(p_2))} \Rightarrow \om{\post{p_1}} = p_2$.\label{INsub}
			\item $\forall p_1 \in P_1 \cap \varphi^{-1}(p_2)\colon$ there is a sequential component $N' = (P', T',$$ F')$ of $N_1$, such that $p_1 \in P'$, $\varphi^{-1}(\neighb{p_2}) \subseteq T'$.\label{scomp}
		\end{enumerate} 
	\end{enumerate}
\end{definition}

Thus, an $\widehat{\alpha}$-morphism is an $\alpha$-morphism (see Definition~\ref{alpham}) that also satisfies conditions 2' and 3' of Definition \ref{alphahm}.
When two LGWF-nets are related by an $\widehat{\alpha}$-morphism, their underlying GWF-nets are related by an $\alpha$-morphism.
That is why $\widehat{\alpha}$-morphisms inherit the properties of $\alpha$-morphisms, discussed in Section \ref{ssec:prop}. We use them to achieve soundness preservation in the AS-composition of LGWF-nets.

Moreover, it also follows from Definition \ref{alphahm} that labeled places in LGWF-nets are both preserved and reflected by $\widehat{\alpha}$-morphisms.
In other words, an image of a labeled place in $N_1$ is a labeled place in $N_2$, as well as an inverse image of a labeled place in $N_2$, is a labeled place in $N_1$.
Thus, there is a bijection between the sets of labeled places in two LGWF-nets related by an $\widehat{\alpha}$-morphism.
\smallskip

We next discuss our approach to ensuring that the AS-composition of sound LGWF-nets yields a sound LGWF-net.
Given two sound LGWF-nets $R_1$ and $R_2$, we aim to be sure that $R_1\! \circledast\!R_2$ is sound.
It is possible to compose $R_1$ and $R_2$ using Definition \ref{comp}, but their composition may not be sound, as shown in the previous section.
A technique described below is applied to achieve soundness of $R_1\! \circledast \!R_2$ by construction.

We start with abstracting $R_1$ and $R_2$ preserving labeled transitions.
Abstractions of LGWF-nets can be constructed by applying finite sequences of transformations, as we discussed in \cite{Pnse-20}.
Thus, we obtain two abstract LGWF-nets $N_1$ and $N_2$, such that there is an $\widehat{\alpha}$-morphism $\varphi_i \colon R_i \to N_i$ with $i=1,2$.
According to Theorem \ref{spres}, $N_1$ and $N_2$ are sound. 
These abstract models $N_1$ and $N_2$ are then composed by adding the same channels and synchronizing transitions with the same labels, as in $R_1$ and $R_2$.
Correspondingly, $N_1\! \circledast\! N_2$ is an interface pattern that describes interactions between LGWF-nets $R_1$ and $R_2$.
We \emph{verify} the soundness and structural properties of $N_1\! \circledast \!N_2$.

Given the sound interface pattern $N_1\! \circledast\! N_2$ and two $\widehat{\alpha}$-morphisms $\varphi_i \colon R_i \to N_i$ with $i=1, 2$, we construct two new LGWF-nets $R_1\! \circledast\! N_2$ and $N_1\! \circledast \!R_2$ representing two \emph{intermediate} refinements of the same abstract interface.
It is easy to see that these refinements of $N_1\! \circledast \!N_2$ preserve $\widehat{\alpha}$-morphisms, i.e., there is an $\widehat{\alpha}$-morphism from $R_1\! \circledast \!N_2$ to $N_1\! \circledast \!N_2$ as well as from $N_1\! \circledast \!R_2$ to $N_1\! \circledast\! N_2$.
For instance, an $\widehat{\alpha}$-morphism from $R_1\! \circledast \!N_2$ to $N_1 \!\circledast \!N_2$ is constructed from the original $\widehat{\alpha}$-morphism $\varphi_1 \colon R_1 \to N_1$ together with an identity mapping of asynchronously labeled transitions in $N_2$ and a corresponding mapping of synchronized transitions that can be refined in $R_1$.
Symmetrically, it is possible to show the construction of an $\widehat{\alpha}$-morphism from $N_1\! \circledast \!R_2$ to $N_1\! \circledast \!N_2$.

In Proposition \ref{irmor}, we additionally claim that an $\widehat{\alpha}$-morphism from an intermediate refinement $R_1 \circledast N_2$ to an interface pattern $N_1 \circledast N_2$ \emph{reflects} the connections among asynchronously labeled transitions with channels --- labeled places.
This reflection follows from the fact that labeled places are both preserved and reflected by $\widehat{\alpha}$-morphisms.
Further, we will use this property in the proof of the main theorem.

\begin{proposition}\label{irmor}
	Let $R_1, N_1, N_2$ be three LGWF-nets, such that there is an $\widehat{\alpha}$-morphism $\varphi_1 \colon R_1 \to N_1$.
	Let $N_1 \!\circledast \!N_2 = (P, T, F, m_0, m_f, h, \ell, k)$ and $R_1\! \circledast \!N_2 = (P', T', F',$ $m_0', m_f', h', \ell', k')$.
	Then there is an $\widehat{\alpha}$-morphism $\varphi_1' \colon (R_1 \!\circledast \!N_2) \to (N_1\! \circledast \!N_2)$, where $\forall p \in \dom{k}$ and $\forall t \in T \colon$
	\begin{enumerate}
		\item if $(p, t) \in F$, then $\{\omr{p}\} \times \omr{t} \subseteq F'$;
		\item if $(t, p) \in F$, then $\omr{t} \times \{\omr{p}\} \subseteq F'$.
	\end{enumerate}
\end{proposition}

Let $N_1 \circledast N_2$ shown in Fig.\,\ref{comp_ex2} represent an interface pattern.
We refine it with two component LGWF-nets $R_1$ and $R_2$, as shown in Fig.\,\ref{intref}.
The corresponding $\widehat{\alpha}$-morphisms are indicated by the shaded ovals.
The $\alpha$-morphism between the underlying GWF-nets $R_2^-$ and $N_2^-$ is provided in Fig. \ref{alpha_ex}, where $N_1$ corresponds to $R_2^-$, and $N_2$ corresponds to $N_2^-$.
\begin{figure}[h]
	\centering
	\begin{subfigure}[b]{0.5\textwidth}
		\centering
		\includegraphics[height=6.5cm]{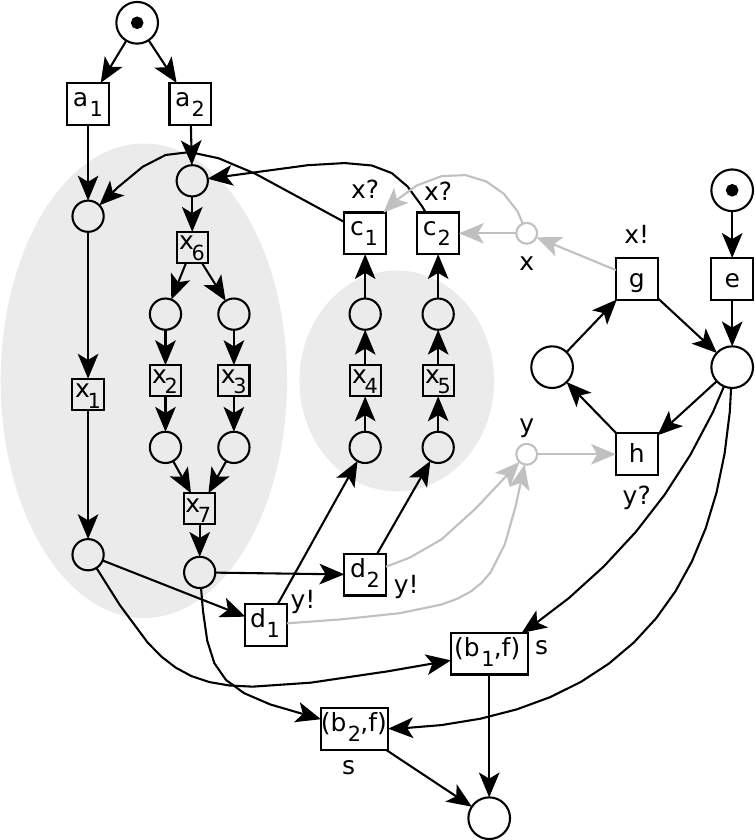}
		\caption{AS-composition $R_1 \circledast N_2$}
	\end{subfigure}%
	\begin{subfigure}[b]{0.5\textwidth}
		\centering
		\includegraphics[height=6.5cm]{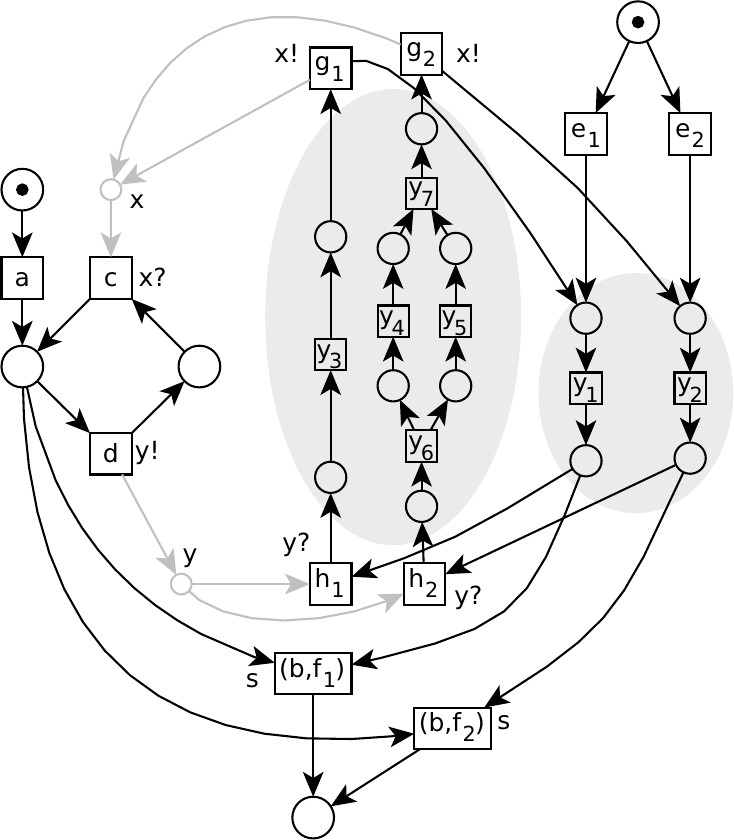}
		\caption{AS-composition $N_1 \circledast R_2$}
	\end{subfigure}
	\caption{Two intermediate refinements of $N_1 \circledast N_2$ from Fig. \ref{comp_ex2}} \label{intref}
\end{figure}

Theorem \ref{mainth} expresses the main result of our study. 
We prove that an $\widehat{\alpha}$-morphism from an intermediate refinement $R_1 \circledast N_2$ (symmetrically, from $N_1 \circledast R_2$) to an  interface pattern $N_1 \circledast N_2$ \emph{reflects} its soundness.
In proving this fact, we use the properties of $\alpha$-morphisms discussed in Section \ref{sec:alpha}, the characterization of reachable markings in the AS-composition given in Proposition \ref{markdec}, and the property considered in Proposition \ref{irmor}.

\begin{theorem}\label{mainth}
	Let $R_1, N_1, N_2$ be three sound LGWF-nets, such that there is an $\widehat{\alpha}$-morphism $\varphi_1 \colon$ $R_1\!\to\!N_1$.
	If $N_1 \circledast N_2$ is sound, then $R_1 \circledast N_2$ is sound.
\end{theorem}
\begin{proof}
	By Proposition \ref{irmor}, there is an $\widehat{\alpha}$-morphism $\varphi_1' \colon (R_1\! \circledast\! N_2) \to (N_1\! \circledast\! N_2)$.
	
	We first fix a notation used in the proof.
	Let $N_i = (P_i, T_i, F_i, m_0^i, m_f^i, h_i, \ell_i,$ $k_i)$, where $i=1, 2$, and $R_1 = (\underline{P}_1, \underline{T}_1, \underline{F}_1, \underline{m}_0^1,$ $\underline{m}_f^1, \underline{h}_1, \underline{\ell}_1, \underline{k}_1)$.
	Also, let $N_1\! \circledast\! N_2 = (P, T, F, m_0, m_f, h, \ell, k)$, and $R_1\! \circledast\! N_2 = (P', T', F', m_0', m_f', h', \ell', k')$.
	
	We show that $R_1\! \circledast\! N_2$ satisfies the three behavioral conditions of a sound LGWF-net imposed by Definition \ref{sound}.
	
	\textbf{1.} Take $m' \in \reach{m_0'}$.
	By Proposition \ref{markdec} for $R_1 \circledast N_2$, $m' = (\underline{m}_1 \setminus \dom{\underline{k}_1}) \cup (m_2 \setminus \dom{k_2}) \cup m_c$, where $\underline{m}_1 \in \reach{\underline{m}_0^1}$, $m_2 \in \reach{m_0^2}$ and $m_c \in \dom{k'}$.
	By Proposition \ref{markpres} for $\varphi_1'$, $\varphi_1'(m') = m \in \reach{m_0}$.
	By Proposition \ref{markdec} for $N_1 \circledast N_2$, $m = (m_1 \setminus \dom{k_1}) \cup (m_2 \setminus \dom{k_2}) \cup m_c$, where $m_2 \setminus \dom{k_2}$, $m_c$ are the same as in $m'$, and $m_1 = \varphi_1(\underline{m}_1)$ (by Proposition \ref{markpres} for $\varphi_1$).
	Since $N_1\! \circledast\! N_2$ is sound, $\exists w \in F\!S(N_1\! \circledast\! N_2) \colon m\reach{w}m_f$.
	By Definition \ref{comp}, recall that $T = T_1^a \cup T_2^a \cup T_{sync}$ in $N_1\! \circledast\! N_2$, where $T_{sync} = \{(t_1, t_2) \, \vert \, t_1 \in \dom{\ell_1}, t_2 \in \dom{\ell_2}, \text{ and } \ell_1(t_1) = \ell_2(t_2) \}$ and $T_i^a = T_i \setminus \dom{\ell_i}$ with $i=1, 2$.
	Using interleaving semantics for Petri nets, we can write $w = w_2^1v$, such that $v=\varepsilon$ or $v=t_1^1w_s^1w_2^2t_1^2\dots$, where $w_2^i \in (T_2^a)^*$, $t_1^i \in T_1^a$ and $w_s^i \in T_{sync}^*$ with $i \geq 1$.
	Firstly, each sub-sequence $w_2^i$ can be obviously simulated on the LGWF-net $N_2$ in $R_1\! \circledast N_2$, since $\varphi_1'$ reflects connections with labeled places (by Proposition \ref{irmor}).
	Secondly, since $R_1$ is sound, $\varphi_1$ reflects reachable markings and transitions firings (by Proposition \ref{markrefl}).
	Thus, there is a reachable marking $\underline{m}_1^i$ in $R_1$, belonging to $\varphi_1^{-1}(m_1^i)$ for some $m_1^i \in \reach{m_0^1}$ in $N_1$.
	If $m_1^i \reach{t_1^i}$ in $N_1$, then $\underline{m}_1^i$ enables all transitions in $\varphi_1^{-1}(t_1^i)$ in $R_1$ as well.
	Moreover, these transitions are also enabled in $R_1\! \circledast\! N_2$, since $\varphi_1'$ reflects connections to labeled places (by Proposition \ref{irmor}).
	Finally, since $N_1\! \circledast\! N_2$ is sound, $\exists m \in \reach{m_0} \colon m \reach{(t_1, t_2)}$ for all $(t_1, t_2)$ in $w_s^i$.
	By Proposition \ref{markdec}, $m = m_1 \cup m_2$, where $m_1 \in \reach{m_0^1}$ and $m_2 \in \reach{m_0^2}$ (here $m_c = \varnothing$, since transitions in $T_{sync}$ are not connected with labeled places).
	Moreover, $m_1 \reach{t_1}$ and $m_2 \reach{t_2}$.
	By Proposition \ref{markrefl} for $\varphi_1$, there is a reachable marking $\underline{m}_1'$ in $R_1$, such that $\underline{m}_1' = \varphi_1^{-1}(m_1)$ and $\forall \underline{t}_1 \in \varphi_1^{-1}(t_1) \colon \underline{m}_1' \reach{\underline{t}_1}$.
	Correspondingly, a reachable marking $\underline{m}_1' \cup m_2$ in $R_1\! \circledast\! N_2$ enables synchronized transitions $(\underline{t}_1, t_2)$ for all $\underline{t}_1 \in \varphi_1^{-1}(t_1)$.
	Hence, we reflect the complete firing sequence $w \in F\!S(N_1\! \circledast\! N_2)$ on $R_1 \circledast N_2$ reaching its final marking $m_f'$.
	
	\textbf{2.} Suppose by contradiction $\exists m' \in \reach{m_0'} \colon m' \supseteq m_f'$ and $m' \neq m_f'$.
	By Definition~\ref{comp}.2, $m_f' = \underline{m}_f^1 \cup m_f^2$.
	Thus, $m' = \underline{m}_f^1 \cup m_f^2 \cup m_3$.
	By Proposition \ref{markpres} for $\varphi_1'$, we have that $\varphi_1'(m') \in \reach{m_0}$.
	Then, $\varphi_1'(m') = \varphi_1'(\underline{m}_f^1) \cup \varphi_1'(m_f^2) \cup \varphi_1'(m_3) = \varphi_1(\underline{m}_f^1) \cup m_f^2 \cup m_3 = m_f^1 \cup m_f^2 \cup m_3 = m_f \cup m_3$.
	This reachable marking $m_f \cup m_3$ strictly covers the final marking $m_f$ in $N_1\! \circledast\! N_2$, which contradicts the assumption of its soundness.
	
	\textbf{3.} We show that $\forall t' \in T'\,\exists m' \in \reach{m_0'} \colon m' \reach{t'}$.
	By Proposition \ref{markdec}, $m' = (\underline{m}_1 \setminus \dom{\underline{k}_1}) \cup (m_2 \setminus \dom{k_2}) \cup m_c$, where $\underline{m}_1 \in \reach{\underline{m}_0^1}$, $m_2 \in \reach{m_0^2}$ and $m_c \in \dom{k'}$.  
	By Definition \ref{comp}.3, $\forall t' \in T' \colon t' \in\underline{T}_1^a$ or $t' \in T_2^a$ or $t' \in T_{sync}$.
	If $t' \in T_2^a$, then $\exists m \in \reach{m_0} \colon m\reach{t'}$, since $N_1\! \circledast\! N_2$ is sound.
	By Proposition \ref{irmor}, $(m_2 \setminus \dom{k_2}) \cup m_c$ in $R_1\! \circledast\! N_2$ also enables $t'$.
	If $t' \in \underline{T}_1^a$, then there are two cases.
	If $\varphi_1'(t') \in P$, then $t'$ is not connected to labeled places.
	Since $R_1$ sound, $\underline{m}_1$ enables $t'$.
	If $\varphi_1'(t') \in T$, then take $t \in T$, such that $\varphi_1'(t') = t$ (by the surjectivity of $\varphi_1'$).
	Since $N_1\! \circledast\! N_2$ is sound, $\exists m \in \reach{m_0} \colon m \reach{t}$.
	By Proposition \ref{markrefl} and \ref{irmor}, the reachable marking $\underline{m}_1 \cup m_c$ in $R_1\! \circledast\! N_2$ (being the inverse image of $m$ under $\varphi_1'$) enables $t'$.
	As for the case when $t' \in T_{sync}$, we have already considered it above when proving reachability of the final marking in $R_1\! \circledast\! N_2$. \qed
\end{proof}

Having two $\widehat{\alpha}$-morphisms from the intermediate refinements $R_1\! \circledast\!N_2$ and $N_1\!\circledast\!R_2$ to the same abstract interface $N_1 \!\circledast\! N_2$, we can compose $R_1 \!\circledast \!N_2$ and $N_1\! \circledast\! R_2$ using the composition defined in~\cite{Bernardinello2013}.
It is required to (a) substitute subnets in $R_1\!\circledast\!N_2$ and $N_1\! \circledast\!R_2$ for the corresponding places in $N_1\!\circledast\!N_2$;
(b)~replace transitions in $N_1\!\circledast\!N_2$ with their inverse images merging those with identical images.
As a result, we obtain $N$ and two $\widehat{\alpha}$-morphisms from $N$ to  $R_1\!\circledast\!N_2$ and $N_1\!\circledast\!R_2$, such that the diagram shown in Fig.\,\ref{diagram1} commutes, i.e., $\varphi_1' \circ \varphi_1'' = \varphi_2' \circ \varphi_2''$, where $\varphi_1' \colon (R_1\!\circledast\! N_2) \to (N_1 \!\circledast\! N_2)$, $\varphi_2' \colon (N_1 \!\circledast\! R_2) \to (N_1 \!\circledast\! N_2)$, $\varphi_1'' \colon N \to (R_1 \!\circledast\! N_2)$, and $\varphi_2'' \colon N \to (N_1 \!\circledast\! R_2)$.

\begin{figure}[h]
	\centering
	\begin{subfigure}[b]{0.5\textwidth}
		\centering
		\includegraphics[height=6.5cm]{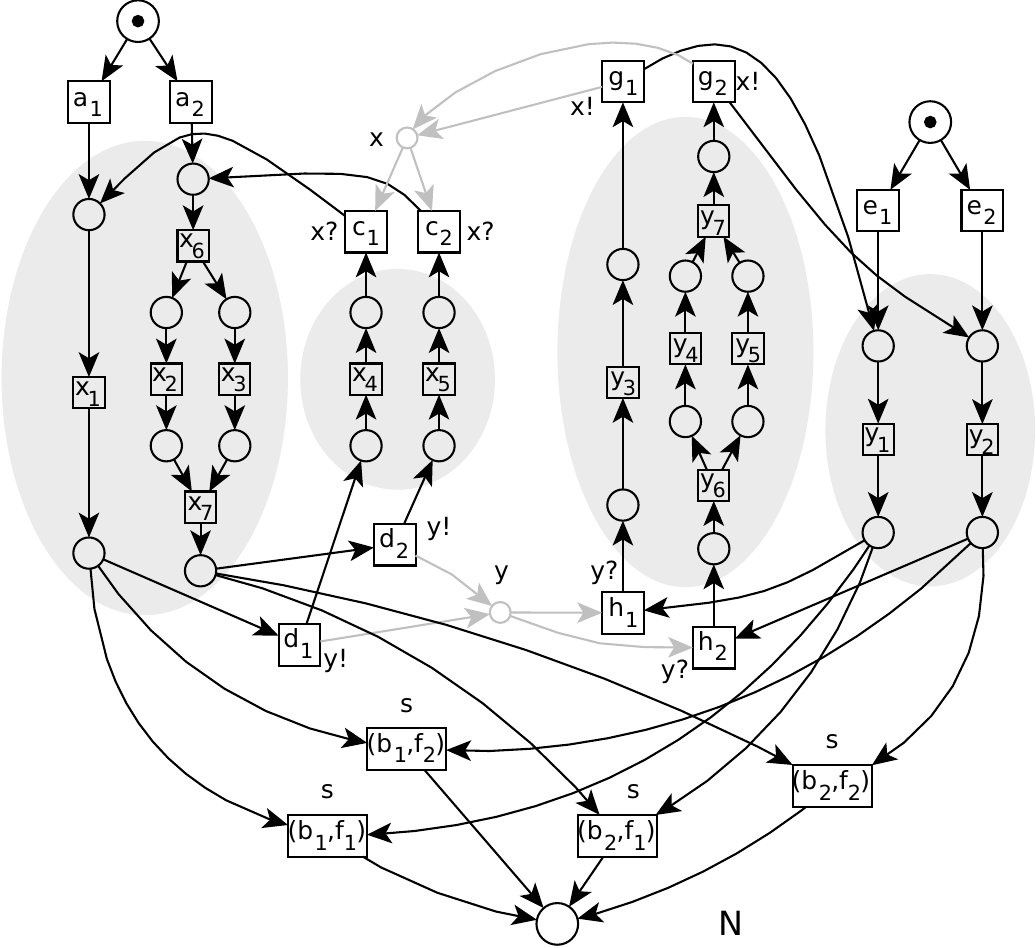}
		\caption{$N$, isomorphic to $R_1 \circledast R_2$\label{comp_res}}
	\end{subfigure}%
	\begin{subfigure}[b]{0.5\textwidth}
		\centering
		\includegraphics[height=5.5cm]{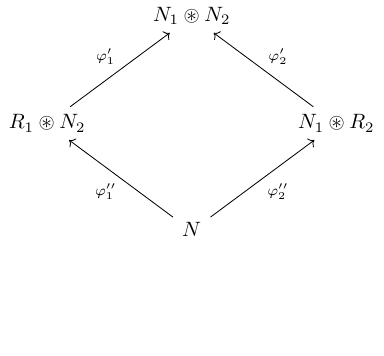}
		\caption{diagram\label{diagram1}}
	\end{subfigure}
	\caption{Composition of $R_1 \circledast N_2$ and $N_1 \circledast R_2$ based on $\widehat{\alpha}$-morphisms}
\end{figure}

Another way is to construct intermediate refinements again by refining $N_2$ in $R_1\! \circledast\! N_2$ ($N_1$ in $N_1\! \circledast \!R_2$).
As a result, we obtain $R_1 \circledast R_2$, isomorphic to the previously constructed composition $N$ up to renaming of synchronized transitions.
According to Proposition \ref{irmor}, there are two $\widehat{\alpha}$-morphisms from $R_1 \circledast R_2$ to $R_1 \circledast N_2$ as well as to $N_1 \circledast R_2$. 
According to Theorem \ref{mainth}, since $R_1 \circledast N_2$ ($N_1 \circledast R_2$) is sound, $R_1 \!\circledast\! R_2$ is also sound.
Therefore, we have also shown that it is possible to simultaneously refine $N_1$ and $N_2$ in the sound abstract interface with sound LGWF-nets $R_1$ and $R_2$ (see Corollary~\ref{cor:main}).

In Fig.\,\ref{comp_res}, we show the result of composing, by means of $\widehat{\alpha}$-morphisms, intermediate refinements $R_1 \circledast N_2$ and $N_1 \circledast R_2$ provided in Fig.\,\ref{intref}. 
This composition corresponds to the direct AS-composition of $R_1$ and $R_2$.

\begin{corollary}\label{cor:main}
	Let $R_1, R_2, N_1, N_2$ be four sound LGWF-nets, such that there is an $\widehat{\alpha}$-morphism $\varphi_i \colon R_i  \to N_i$ for $i=1, 2$. 
	If $N_1\!\circledast\! N_2$ is sound, then $R_1\! \circledast \!R_2$ is sound.
\end{corollary}

\section{Related Works}\label{sec:relw}

There is a considerable amount of literature devoted to the compositional modeling of Petri nets, including, among the others, \cite{Kotov78,BoxCalc,Reisig13,Valk2003}.
The recent works \cite{Reisig2018,ReiComp20,Her2021} by Wolfgang Reisig are devoted to a systematic study of compositional modeling principles applicable to various formalisms and notations, including (Colored) Petri nets, BPMN (Business Process Modeling and Notation) process models, and UML (Unified Modeling Language) diagrams.

Several works studied whether a composition of \emph{open} Petri nets preserves semantical properties of components.
Paolo Baldan et al.~\cite{Baldan01} introduced a class of open Petri nets and analyzed the categorical framework behind open Petri net composition constructed via place and transition fusion. 
Kees van Hee et al.~\cite{vanHee2010,vanHee2011} considered a soundness-preserving refinement of places in open WF-nets with sound (composition of) WF-nets.

A class of superposed automata nets (SA-nets) was introduced by Fiorella De Cindo et al.~in \cite{SANets82}.
SA-nets were among the first formalisms to model systems with synchronously communicating sequential components via transition fusion.

Serge Haddad et al.~\cite{Haddad13} defined the semantics of \emph{input/output} (I/O) Petri nets and their composition constructed through the insertion of asynchronous channels. 
The authors studied channel properties related to message consumption and interaction termination.
It was shown that these properties are decidable and preserved by an asynchronous composition of I/O-Petri nets. 

Younes Soussi and G\'{e}rard Memmi \cite{Sous91,SousMem1} considered the problem of liveness preservation in a composition of Petri nets through an intermediate model of communication medium. 
Their approach is based on global and rigid structural constraints.

Christian Stahl and Karsten Wolf \cite{Wolf09} applied \emph{operating guidelines} for compositional verification of deadlock-freeness in the composition of open Petri nets. 
Their work also considered a problem to decide if one can replace a component in a composition preserving its semantical properties.

Inheritance of behavioral properties of Petri nets is also achieved with the help of morphisms~--- structural graph mappings. 
The composition of Petri nets via morphisms was a subject of many works, including, for example, \cite{Winskel67,monoids90,Nielsen92,Bednarczyk03,Padberg2003,Fabre06,Nhat,Desel2010}. 
We note that morphisms provide a natural and rigid framework to explore properties of Petri net composition.

In our study, the soundness preservation in the AS-composition of LGWF-nets is achieved with the help of a restriction of $\alpha$-morphisms, originally defined by Luca Bernardinello et al. in \cite{Bernardinello2013}. 
They allow us to abstract subnets and refine places in Petri nets. In addition, $\alpha$-morphisms preserve and reflect reachable markings and induce the bisimulation between related models.

Several works have discussed architectural and semantical aspects of compositional approaches to WF-net modeling.
Juliane Siegeris and Armin Zimmermann \cite{WFRes} considered specific patterns of WF-net interactions preserving relaxed soundness of components admitting executions that may not terminate in a final state.
The work \cite{Lomazova13} by Irina Lomazova and Ivan Romanov addressed the problem of preserving service correctness in the context of resources produced and consumed by interacting services. 
The earlier work \cite{Lomazova10} by I.\,Lomazova also proposed an approach to soundness-preserving re-engineering of hierarchical WF-nets with a two-level structure.

Yudith Cardinale et al., in the survey \cite{Cardinale13}, discussed a variety of approaches to compositional modeling of web services. 
The authors stressed that there is a lack of service execution techniques based on different classes of Petri nets.
In particular, Victor Pancratius and Wolffried Stucky \cite{Pankratius05}~considered the composition of WF-nets representing web service behavior with the help of adapted relational algebra operations.

The main difference in our work is that the AS-composition of labeled GWF-nets leaves asynchronous channels and synchronous transitions open for other components to connect.
Apart from that, refinement of LGWF-nets is defined at the level of a complete net rather than specific places and transitions.
Refinement preserves the soundness of LGWF-net components and an interface, which describes interactions at the abstract level.

In our earlier paper \cite{ATAED-18}, we discussed a restricted case of modeling semantically correct asynchronous communication of workflow nets.
Results presented in this paper naturally extend the previous ones and provide the formal backgrounds for constructing sound workflow nets from sound models of interacting components.

\section{Conclusions}\label{sec:concl}
This paper has studied the theoretical backgrounds for a correct composition of interacting workflow net components.
We have developed an approach to model asynchronous and synchronous interactions among workflow nets using two kinds of transition labels.
Correspondingly, we have defined an asynchronous-synchronous composition (AS-composition). 
AS-composition may not preserve the soundness of interacting components.
To solve this problem, we use an interface that describes how WF-nets interact.

An interface net represents an abstract view of a complete system. 
There is a subnet in an interface corresponding to component behavior.
The correspondence between an interface and components is established with the help of $\alpha$-morphisms.
The structural and behavioral properties of the abstraction/refinement relation based on $\alpha$-morphisms have helped us to prove that refining subnets in an interface with sound WF-nets preserves the interface soundness.

We identify two main advantages of the proposed compositional approach.
Firstly, the problem of constructing a correct composition of interacting workflow net components is solved in the abstract model.
Refinement of abstract places requires checking structural constraints and only local behavioral constraints for properly refined places.
Sound models of abstract interfaces can be reused with different component refinements.
Secondly, AS-composition leaves asynchronous channels and synchronous transitions of workflow net components open for others to interact.

Our future research will be focused on the following aspects.
It is planned to relax constraints of $\alpha$-morphisms to enable the abstraction of cyclic subnets.
Thus, an acyclic abstract interface will define a broader class of sound WF-net compositions.
We also plan to systematically identify typical interface patterns that a system architect can use to organize smooth interactions among components in large-scale distributed systems.
Our earlier works \cite{TMPA-17,Macspro-19} studied patterns of asynchronous interactions within a restricted case of two interacting components.



 \bibliographystyle{elsarticle-num} 
\bibliography{mybibfile-noDOI}

\begin{thebibliography}{10}
\expandafter\ifx\csname url\endcsname\relax
  \def\url#1{\texttt{#1}}\fi
\expandafter\ifx\csname urlprefix\endcsname\relax\def\urlprefix{URL }\fi
\expandafter\ifx\csname href\endcsname\relax
  \def\href#1#2{#2} \def\path#1{#1}\fi

\bibitem{Reisig13}
W.~Reisig, Understanding Petri Nets: Modeling Techniques, Analysis Methods,
  Case Studies, Springer, 2013.

\bibitem{Reisig2018}
W.~Reisig, Associative composition of components with double-sided interfaces,
  Acta Informatica 56 (2019) 229--253.

\bibitem{ReiComp20}
W.~Reisig, Composition of component models -- a key to construct big systems,
  in: Leveraging Applications of Formal Methods, Verification and Validation:
  Engineering Principles, Vol. 12477 of Lecture Notes in Computer Science,
  Springer International Publishing, 2020, pp. 171--188.

\bibitem{Wolf09}
C.~Stahl, K.~Wolf, Deciding service composition and substitutability using
  extended operating guidelines, Data \& Knowledge Engineering 68 (2009)
  819--833.

\bibitem{Aalstwf02}
W.~van~der Aalst, Workflow verification: Finding control-flow errors using
  petri-net-based techniques, in: Business Process Management, Vol. 1806 of
  Lecture Notes in Computer Science, Springer Heidelberg, 2020, pp. 161--183.

\bibitem{Bernardinello2013}
L.~Bernardinello, E.~Mangioni, L.~Pomello, Local state refinement and
  composition of elementary net systems: An approach based on morphisms, in:
  ToPNoC VIII, Vol. 8100 of LNCS, Springer, Heidelberg, 2013, pp. 48--70.

\bibitem{Rozenberg96}
G.~Rozenberg, J.~Engelfriet, Elementary net systems, in: Lectures on Petri Nets
  I: Basic Models, Vol. 1491 of LNCS, Springer, Heidelberg, 1998, pp. 12--121.

\bibitem{Aalst11}
W.~van~der Aalst, K.~van Hee, A.~ter Hofstede, N.~Sidorova, H.~Verbeek,
  M.~Voorhoeve, M.~Wynn, Soundness of workflow nets: classification,
  decidability, and analysis, Formal Aspects of Computing 23~(3) (2011)
  333--363.

\bibitem{Pnse-20}
L.~Bernardinello, I.~Lomazova, R.~Nesterov, L.~Pomello, Property-preserving
  transformations of elementary net systems based on morphisms, in:
  Transactions on Petri Nets and Other Models of Concurrency (ToPNoC) XVI, Vol.
  13220 of Lecture Notes in Computer Science, Springer, 2022, pp. 1--23.

\bibitem{Kotov78}
V.~Kotov, An algebra for parallelism based on petri nets, in: MFCS 1978,
  Vol.~64 of Lecture Notes in Computer Science, Springer, Heidelberg, 1978, pp.
  39--55.

\bibitem{BoxCalc}
E.~Best, R.~Devillers, J.~Hall, The box calculus: A new causal algebra with
  multi-label communication, in: Advances in Petri Nets 1992, Vol. 609 of LNCS,
  Springer, 1992, pp. 21--69.

\bibitem{Valk2003}
C.~Girault, R.~Valk, Petri Nets for Systems Engineering: A Guide to Modeling,
  Verification, and Applications, Springer, Heidelberg, 2003.

\bibitem{Her2021}
P.~Fettke, W.~Reisig, Modelling service-oriented systems and~cloud services
  with heraklit, in: Advances in Service-Oriented and Cloud Computing, Vol.
  1360 of Communications in Computer and Information Science, Springer
  International Publishing, Cham, 2021, pp. 77--89.

\bibitem{Baldan01}
P.~Baldan, A.~Corradini, H.~Ehrig, R.~Heckel, Compositional modeling of
  reactive systems using open nets, in: CONCUR 2001, Vol. 2154 of LNCS,
  Springer, 2001, pp. 502--518.

\bibitem{vanHee2010}
K.~van Hee, N.~Sidorova, J.~van~der Werf, Construction of asynchronous
  communicating systems: Weak termination guaranteed!, in: Software
  Composition, Vol. 6144 of LNCS, Springer, Heidelberg, 2010, pp. 106--121.

\bibitem{vanHee2011}
K.~M. van Hee, A.~J. Mooij, N.~Sidorova, J.~van~der Werf, Soundness-preserving
  refinements of service compositions, in: Web Services and Formal Methods,
  Vol. 6551 of LNCS, Springer, Heidelberg, 2011, pp. 131--145.

\bibitem{SANets82}
F.~De~Cindo, G.~De~Michelis, L.~Pomello, C.~Simone, Superposed automata nets,
  in: Application and Theory of Petri Nets, Vol.~52 of Informatik-Fachberichte,
  Springer, Heidelberg, 1982, pp. 269--279.

\bibitem{Haddad13}
S.~Haddad, R.~Hennicker, M.~M{\o}ller, Channel properties of asynchronously
  composed petri nets, in: ICATPN 2013, Vol. 7927 of LNCS, Springer,
  Heidelberg, 2013, pp. 369--388.

\bibitem{Sous91}
Y.~Souissi, On liveness preservation by composition of nets via a set of
  places, in: Advances in Petri Nets 1991, Vol. 524 of LNCS, Springer,
  Heidelberg, 1991, pp. 277--295.

\bibitem{SousMem1}
Y.~Souissi, G.~Memmi, Composition of nets via a communication medium, in:
  Advances in Petri Nets 1990, Vol. 483 of LNCS, Springer, Heidelberg, 1991,
  pp. 457--470.

\bibitem{Winskel67}
G.~Winskel, Petri nets, morphisms and compositionality, in: Advances in Petri
  Nets 1985, Vol. 222 of LNCS, Springer, Heidelberg, 1986, pp. 453--477.

\bibitem{monoids90}
J.~Meseguer, U.~Montanari, Petri nets are monoids, Information and Computation
  88~(2) (1990) 105--155.

\bibitem{Nielsen92}
M.~Nielsen, G.~Rozenberg, P.~Thiagarajan, Elementary transition systems,
  Theoretical Computer Science 96~(1) (1992) 3--33.

\bibitem{Bednarczyk03}
M.~Bednarczyk, L.~Bernardinello, B.~Caillaud, W.~Paw{\l}owski, L.~Pomello,
  Modular system development with pullbacks, in: ICATPN 2003, Vol. 2679 of
  LNCS, Springer, Heidelberg, 2003, pp. 140--160.

\bibitem{Padberg2003}
J.~Padberg, M.~Urb{\'a}{\v{s}}ek, Rule-based refinement of petri nets: A
  survey, in: Petri Net Technology for Communication-Based Systems: Advances in
  Petri Nets, Vol. 2472 of LNCS, Springer, Heidelberg, 2003, pp. 161--196.

\bibitem{Fabre06}
E.~Fabre, On the construction of pullbacks for safe petri nets, in: ICATPN
  2006, Vol. 4024 of LNCS, Springer, Heidelberg, 2006, pp. 166--180.

\bibitem{Nhat}
L.~Bernardinello, E.~Monticelli, L.~Pomello, On preserving structural and
  behavioural properties by composing net systems on interfaces, Fundamenta
  Informaticae 80~(1-3) (2007) 31--47.

\bibitem{Desel2010}
J.~Desel, A.~Merceron, Vicinity respecting homomorphisms for abstracting system
  requirements, in: ToPNoC IV, Vol. 6550 of LNCS, Springer, Heidelberg, 2010,
  pp. 1--20.

\bibitem{WFRes}
J.~Siegeris, A.~Zimmermann, Workflow model compositions preserving relaxed
  soundness, in: BPM 2006, Vol. 4102 of LNCS, Springer, 2006, pp. 177--192.

\bibitem{Lomazova13}
I.~Lomazova, I.~Romanov, Analyzing compatibility of services via resource
  conformance, Fundamenta Informaticae 128~(1-2) (2013) 129--141.

\bibitem{Lomazova10}
I.~Lomazova, Interacting workflow nets for workflow process re-engineering,
  Fundamenta Informaticae 101~(1-2) (2010) 59--70.

\bibitem{Cardinale13}
Y.~Cardinale, J.~El~Haddad, M.~Manouvrier, M.~Rukoz, Web service composition
  based on petri nets: Review and contribution, in: RED 2012, Vol. 8194 of
  LNCS, Springer, Heidelberg, 2013, pp. 83--122.

\bibitem{Pankratius05}
V.~Pankratius, W.~Stucky, A formal foundation for workflow composition,
  workflow view definition, and workflow normalization based on petri nets, in:
  APCCM 2005, vol. 43, Australian Computer Society, Inc., 2005, pp. 79--88.

\bibitem{ATAED-18}
L.~Bernardinello, I.~Lomazova, R.~Nesterov, L.~Pomello, Compositional discovery
  of workflow nets from event logs using morphisms, in: Proceedings of
  ATAED-2018, Vol. 2115 of CEUR Workshop Proceedings, CEUR-WS.org, 2018, pp.
  23--38.

\bibitem{TMPA-17}
R.~Nesterov, I.~Lomazova, Compositional process model synthesis based on
  interface patterns, in: TMPA 2017, Vol. 779 of CCIS, Springer, 2018, pp.
  151--162.

\bibitem{Macspro-19}
R.~Nesterov, I.~Lomazova, Asynchronous interaction patterns for mining
  multi-agent system models from event logs, in: Proceedings of MACSPro-2019,
  Vol. 2478 of CEUR Workshop Proceedings, CEUR-WS.org, 2019, pp. 1--12.

\end{thebibliography}

%
%
%
%

\end{document}